\documentclass[]{interact}

\usepackage{amssymb, amsmath, amsthm, graphicx, epstopdf, booktabs, xcolor}
\usepackage{epstopdf}
\usepackage{subfigure}
\usepackage{bm} 
\usepackage{mathrsfs} 
\usepackage{commath} 

\usepackage{hyperref}

\usepackage{natbib}
\bibpunct[, ]{(}{)}{;}{a}{}{,}

\allowdisplaybreaks[4]

\theoremstyle{plain}
\newtheorem{thm}{Theorem}[section]
\newtheorem{lem}[thm]{Lemma}

\newtheorem{prop}[thm]{Proposition}

\theoremstyle{definition}

\newtheorem{asp}[thm]{Assumption}

\theoremstyle{remark}
\newtheorem{rem}{Remark}

\newcommand{\pder}[2][]{\frac{\partial#1}{\partial#2}}

\DeclareMathOperator*{\esssup}{ess\,sup}

\newcommand{\calA}{\mathcal{A}}

\newcommand{\calC}{\mathcal{C}}
\newcommand{\calF}{\mathcal{F}}
\newcommand{\calL}{\mathcal{L}}

\newcommand{\calS}{\mathcal{S}}

\newcommand{\scrF}{\mathscr{F}}
\newcommand{\scrL}{\mathscr{L}}

\newcommand{\E}{\mathbb{E}}				
\newcommand{\prob}{\mathbb{P}}		
\newcommand{\Q}{\mathbb{Q}}

\newcommand{\vm}[1]{\mathbf{#1}}

\begin{document}


\title{A Put-Call Transformation of the Exchange Option Problem under Stochastic Volatility and Jump-Diffusion Dynamics}

\author{
\name{Len Patrick Dominic M. Garces\textsuperscript{a,b}\thanks{Corresponding Author; E-mail: len\_patrick\_dominic.garces@mymail.unisa.edu.au; ORCiD: \url{https://orcid.org/0000-0002-2737-7348
}} and Gerald H. L. Cheang\textsuperscript{a}\thanks{E-mail: Gerald.Cheang@unisa.edu.au; ORCiD: \url{https://orcid.org/0000-0003-3786-0285}}}
\affil{\textsuperscript{a}Centre for Industrial and Applied Mathematics, School of Information Technology and Mathematical Sciences, University of South Australia, Mawson Lakes SA 5095, Australia; \textsuperscript{b}Department of Mathematics, School of Science and Engineering, Ateneo de Manila University, Quezon City 1108, Metro Manila, Philippines}
}

\maketitle

\begin{abstract}
We price European and American exchange options where the underlying asset prices are modelled using a \citet{Merton-1976} jump-diffusion with a common \citet{Heston-1993} stochastic volatility process. Pricing is performed under an equivalent martingale measure obtained by setting the second asset yield process as the num\'eraire asset, as suggested by \citet{Bjerskund-1993}. Such a choice for the num\'eraire reduces the exchange option pricing problem, a two-dimensional problem, to pricing a call option written on the ratio of the yield processes of the two assets, a one-dimensional problem. The joint transition density function of the asset yield ratio process and the instantaneous variance process is then determined from the corresponding Kolmogorov backward equation via integral transforms. We then determine integral representations for the European exchange option price and the early exercise premium and state a linked system of integral equations that characterizes the American exchange option price and the associated early exercise boundary. Properties of the early exercise boundary near maturity are also discussed.
\end{abstract}

\begin{keywords}
Exchange options; change-of-num\'eraire; jump diffusion processes; put-call transformation; stochastic volatility
\end{keywords}

\section{Introduction}

We investigate the pricing of European and American exchange options written on assets with prices driven by stochastic volatility and jump-diffusion (SVJD) dynamics. The earliest analysis of European exchange options was that of \citet{Margrabe-1978} who assumed that the underlying non-dividend-paying stocks are modelled with correlated geometric Brownian motions. Assuming that the European exchange option price is linear homogeneous in the stock prices, \citet{Margrabe-1978} transformed the problem to the classical European call option pricing problem and computed the exchange option price using the solution of \citet{BlackScholes-1973}. \citet{Fischer-1978} considered a closely related problem of determining call option prices when the exercise price is also a diffusion process. \citet{Bjerskund-1993} considered the pricing of American exchange options as an optimal stopping problem in a pure diffusion setting. They suggested that by choosing one of the stocks as the num\'eraire and by a change of measure to the corresponding equivalent martingale measure, the American exchange option pricing problem may be simplified to the problem of pricing an American call or put option. \citet{Bjerskund-1993} refer to this technique as the ``put-call transformation.''

With well-established evidence pointing to the deficiencies of the geometric Brownian motion in accurately modelling asset price returns, there has since been a movement to study option prices (including exchange options) under alternative asset price models.\footnote{The empirical literature addressing the limitations of the \citet{BlackScholes-1973} is extremely rich and will not be reviewed in its totality here. Instead, we invite the reader to consult \citet{Bakshi-1997}, \citet{Duffie-2000}, \citet{Cont-2001}, \citet{Andersen-2002}, \citet{Chernov-2003}, \citet{Eraker-2003}, \citet{Kou-2008}, and the references therein.} \citet{CheangChiarellaZiogas-2006} used the put-call transformation technique to price European exchange options where the underlying assets are modelled using correlated \citet{Merton-1976} jump-diffusion models. With the same asset price model specification, \citet{CheangChiarella-2011} and \citet{Caldana-2015} studied European and American exchange option prices using the risk-neutral approach. \citet{PetroniSabino-2018} priced European exchange options assuming that underlying asset price jumps are correlated and considered applications in energy markets. \citet{Antonelli-2010}, \citet{Alos-2017}, and \citet{KimPark-2017} priced European exchange options where underlying assets are driven by stochastic volatility models. Notably, \citet{Alos-2017} also employed the put-call transformation and discussed hedging under the resulting martingale measure. \citet{Fajardo-2006} also used a similar transformation, which they called the ``dual market method'', to price options (including perpetual exchange options) when underlying prices are driven by L\'evy processes. \citet{Ma-2020} analyzed European exchange options when asset prices are modelled using Hawkes jump-diffusion processes, allowing for jump-contagion in individual assets and jump interdependence among multiple assets. More recently, \citet{CheangGarces-2019}, derived analytical representations for the European and American exchange option prices assuming that stock prices are modelled using a pair of \citet{Bates-1996} stochastic volatility and jump-diffusion dynamics. 

In this paper, we derive integral representations of the price of European and American exchange options when underlying asset prices are modelled using \citet{Merton-1976} jump-diffusion dynamics with a common underlying \citet{Heston-1993} stochastic volatility process. To do so, we follow the suggested approach of \citet{Bjerskund-1993} and assign the yield process of the second stock as the num\'eraire asset, in contrast to using the money market account as the num\'eraire. Doing so simplifies the original two-dimensional problem to a one-dimensional problem of pricing an ordinary call option written on the ratio of the yield processes of the two underlying assets.\footnote{While the title of this paper mentions a ``put-call transformation'', our main analysis consists of transforming the exchange option pricing problem to the problem of pricing a call option on the asset yield ratio by choosing the second asset yield process as the num\'eraire. The exchange option problem can be reduced to a put option on the asset yield ratio if one alternatively chooses the yield process of the first asset as the num\'eraire.} This allows us to follow the techniques of \citet{CheangChiarellaZiogas-2013} who analyzed single-asset American call options under the \citet{Bates-1996} stochastic volatility and jump-diffusion dynamics. 

Our main contribution is to extend the \citet{Bjerskund-1993} strategy for valuing American exchange options in a pure-diffusion setting into the SJVD framework. In contrast to \citet{Margrabe-1978} and \citet{Alos-2017}, we place significant added attention to pricing American exchange options using the put-call transformation. Our analysis also extends the study of \citet{CheangChiarellaZiogas-2006} and \citet{CheangChiarella-2011} to consider American exchange options and European exchange options when assets are driven by a stochastic volatility \emph{and} jump-diffusion model. 

In this analysis, pricing takes place under the equivalent martingale measure $\hat{\Q}$ corresponding to setting the second asset yield process as the num\'eraire. Under $\hat{\Q}$, we find that the the no-arbitrage price of the European exchange option can be written as a function of only the asset yield ratio $\tilde{s}$ and the instantaneous variance $v$. Furthermore, we verify an early exercise representation of the discounted American exchange option price, which can also be written as a function of only $\tilde{s}$ and $v$. 

To evaluate expectations under $\hat{\Q}$, we require the joint transition density function of $\tilde{s}$ and $v$. To do so, we determine the Kolmogorov equation for the transition density and solve this using Fourier and Laplace transforms.\footnote{This method has been used by \citet{ChiarellaZiogas-2009} to price American call options under jump-diffusion dynamics, by \citet{ChiarellaZiogasZiveyi-2010} to price American call options under stochastic volatility, by \citet{CheangChiarellaZiogas-2013} to price American call options under SVJD dynamics, and by \citet{ChiarellaZiveyi-2014} to price American spread call options under pure diffusion dynamics.} With the joint transition density function, we then provide integral representations for the European exchange option price and the early exercise premium and present a linked system of integral equations that characterize the price of the American exchange option and the unknown early exercise boundary. We also note that the joint transition density function we obtain in this analysis enables us to price any European-type option written on the two assets provided that its payoff function can be written in terms of the terminal asset yield ratio and instantaneous variance.

Our analysis also serves as an alternative of that of \citet{CheangGarces-2019} in the following aspects. 
\begin{enumerate}
	\item Our SVJD model specification allows us to incorporate the possibility of correlation between the asset price returns and between each individual asset price process and the instantaneous variance process. 
	\item We take one of the asset yield processes as the num\'eraire instead of the money market account used by \citet{CheangGarces-2019}. 
	\item We express option prices in terms of the transition density function of the underlying stochastic processes under $\hat{\Q}$.
	\item We provide a more in-depth analysis of the early exercise boundary (particularly its behavior near maturity) and the early exercise premium for the American exchange option, and thus extending the earlier work of \citet{ChiarellaZiogas-2009} who investigated the limit of the earliy exercise boundary for American call options under jump-diffusion dynamics.
\end{enumerate}

We are primarily concerned with obtaining analytical representations for exchange option prices under our SVJD model. As such, it is not the goal of this paper to discuss the numerical solution of the option pricing problem or the calibration of model parameters of observable market data as these matters warrant their own dedicated exposition.\footnote{There is considerable extant work on calibrating the parameters of stochastic volatility and/or jump-diffusion models in respect to option pricing. These results are discussed in the works cited in the first footnote of this paper.} However, we discuss how our results set the stage for numerical implementation.

Note that while the succeeding analysis focuses on exchange options written on stocks, one may consider exchange options written on other assets such as indices and foreign currencies. For foreign currencies, in particular, the dividend yields are replaced by risk-free interest rates in the domestic and foreign money markets. \citet{Siegel-1995} explains how exchange options can be used to estimate the ``implicit beta'' between an underlying stock and a given market index. The exchange option framework may be adapted to investigate real options \citep{Kensinger-1988, Carr-1995}, outperformance options \citep{CheangChiarella-2011}\footnote{\citet{CheangChiarella-2011} assumed that only one asset price process had jumps while the other was modelled as a pure-diffusion process. \citet{QuittardPinon-2010} discuss in greater detail the European exchange option pricing problem under a similar model specification.}, energy market options \citep[surveyed in][]{Benth-2015}, and the option to enter/exit an emerging market \citep{Miller-2012}, among others. \citet{Ma-2020} provide additional examples of financial contracts which can be priced under the exchange option framework.

The rest of the paper is organized as follows. Section \ref{sec-PutCall-SVJDModel} specifies the stochastic volatility and jump-diffusion model for underlying asset prices, discusses the construction of the measure $\hat{\Q}$, and presents the dynamics of the asset yield ratio and variance processes under $\hat{\Q}$. Section \ref{sec-PutCall-ExchangeOptionIPDE} discusses the integro-partial differential equation (IPDE) for the discounted European exchange option price and the free-boundary IPDE for the American exchange option. Section \ref{sec-PutCall-EER} uses probabilistic arguments to verify the early exercise representation of the American exchange option price and to express the American exchange option price as a solution of an inhomogeneous IPDE to be solved over a domain unrestricted by the early exercise boundary. Section \ref{sec-PutCall-EEBLimit} discusses some properties of the early exercise boundary near the maturity of the option. Section \ref{sec-PutCall-TransitionDensityFunction} shows the solution of the Kolmogorov equation for the joint transition density function using integral transform methods. With the transition density function, Sections \ref{sec-PutCall-EuExcOpPrice} and \ref{sec-PutCall-EEP} present the integral representations for the European and American exchange option prices, respectively. Section \ref{sec-PutCall-Conclusion} concludes the paper. Proofs of some results which involve lengthy, but otherwise rather elementary, calculations are given as appendices.

\section{A Stochastic Volatility Jump-Diffusion Model}
\sectionmark{Underlying Asset Price Model}
\label{sec-PutCall-SVJDModel}

In this section, we discuss the model specification for the underlying stock prices. Let $(\Omega,\calF,\prob)$ be a probability space equipped with a filtration $\{\calF_t\}_{0\leq t\leq T}$ satisfying the usual conditions. Here, $T>0$ represents the maturity date of the exchange option. Let $\{W_1(t)\}$, $\{W_2(t)\}$, and $\{Z(t)\}$ be standard $\prob$-Brownian motions with instantaneous correlations given by $\dif W_1(t)\dif W_2(t) = \rho_w \dif t$ and $\dif W_j(t)\dif Z(t) = \rho_j\dif t$, for $j=1,2.$ Denote by $\bm{\Sigma}$ the correlation matrix of the random vector $\vm{B}(t) = (W_1(t),W_2(t),Z(t))^\top$. Let $p(\dif y_j,\dif t)$ ($j=1,2$) be the counting measure associated to a marked Poisson process with $\prob$-local characteristics $(\lambda_j,m_\prob(\dif y_j))$.\footnote{See \citet{Runggaldier-2003} for more details.} Underlying $p(\dif y_j,\dif t)$ is a sequence of ordered pairs $\{(T_{i,n},Y_{i,n})\}$ where $Y_{i,n}$ is the ``mark'' of the $n$th occurrence of an event that occurs at a non-explosive time $T_{i,n}$. The marks $Y_{j,1},Y_{j,2},\dots$ are i.i.d. real-valued random variables with non-atomic density $m_\prob(\dif y_j)$. Associated to the event times, we define a Poisson counting process $\{N_j(t)\}$ given by $N_j(t) = \sum_{n=1}^\infty \vm{1}(T_{j,n}\leq t)\vm{1}(Y_{j,n}\in\mathbb{R}),$ where $\vm{1}(\cdot)$ is the indicator function. For simplicity, we assume that the intensities $\lambda_j$ and jump-size densities $m_\prob(\dif y_j)$ are constant through time, although the subsequent analysis can be extended to the case where the intensities and densities are deterministic functions of time.

We assume that the counting measures are independent of the Brownian motions and of each other. Henceforth, we assume that $\{\calF_t\}$ is the natural filtration generated by the Brownian motions and the counting measures, augmented with the collection of $\prob$-null sets.

Denote by $\{S_1(t)\}$ and $\{S_2(t)\}$ the price processes of two assets that pay a constant dividend yield of $q_1$ and $q_2$, respectively, per annum. As stock prices may jump, we let $S_1(t)$ and $S_2(t)$ denote the stock prices \emph{prior to any jumps occurring at time $t$}. Let $\{v(t)\}$ be the instantaneous variance process that governs the volatility of both stock price processes. We assume that the dynamics of the stock prices and the instantaneous variance satisfy the stochastic differential equations
\begin{align}
\label{eqn-PutCall-StockPriceSDE-P}
\frac{\dif S_j(t)}{S_j(t)}	& = (\mu_j-\lambda_j\kappa_j)\dif t+\sigma_j \sqrt{v(t)}\dif W_j(t)+\int_{\mathbb{R}}\left(e^{y_j}-1\right)p(\dif y_j,\dif t), \qquad j=1,2,\\
\label{eqn-PutCall-VarianceSDE-P}
\dif v(t) & = \xi\left(\eta-v(t)\right)\dif t + \omega\sqrt{v(t)}\dif Z(t).
\end{align}
Here, $\kappa_j \equiv \E_\prob[e^{Y_j}-1] = \int_{\mathbb{R}}(e^{y_j}-1)m_\prob(\dif y_j)$ is the mean jump size of the price of asset $j$ under $\prob$, and $\mu_j$, $\sigma_j$, $\xi$, $\eta$, and $\omega$ are positive constants. It is also assumed that the initial values of these stochastic processes are positive. We refer to this model as the \emph{proportional stochastic volatility and jump-diffusion (SVJD) model.}

As described above, the model features a common instantaneous variance process and independent jump terms for each asset. The individual jump processes may be taken to model idiosyncratic risk factors in each asset that cause sudden changes in returns.\footnote{In contrast, \citet{CheangChiarella-2011} introduced an additional compound Poisson process appearing in both asset return processes which capture macroeconomic shocks or systematic risk factors which may introduce sudden jumps in returns.} Although extremely rare, it is possible that jumps for both stocks arrive at the same time, representing market shocks or sudden events that may affect both assets. In addition, the common variance process models systematic market volatility or volatility at the macroeconomic level. As such, individual asset prices may provide feedback to each other via the correlation between the diffusion components and the dependence on a common stochastic volatility.

We also assume the existence of a money market account whose value process is denoted by $\{M(t)\}$, with $M(t)=e^{rt}$ for $t\geq 0$, where $r>0$ is the (constant) risk-free interest rate. We require the following assumption on the parameters of the variance process and the correlation parameters to ensure that $\{v(t)\}$ remains strictly positive and finite for all $0\leq t\leq T$ under $\prob$ and any other probability measure equivalent to $\prob$ \citep{AndersenPiterbarg-2007}.

\begin{asp}
\label{asp-PutCall-VarianceParameterAssumptions}
The parameters $\xi$, $\eta$, and $\omega$ and the correlation coefficients $\rho_1$ and $\rho_2$ satisfy $2\xi\eta\geq\omega^2$ and $-1<\rho_j<\min\left\{\xi/\omega,1\right\}$, $j=1,2$.
\end{asp}

Straightforward calculations using It\^o's Lemma for jump-diffusions show that equation \eqref{eqn-PutCall-StockPriceSDE-P} admits a solution of the form $$S_j(t) = S_j(0)\exp\Bigg\{(\mu_j-\lambda_j\kappa_j)t-\frac{1}{2}\sigma_j^2\int_0^t v(s)\dif s+\sigma_2\int_0^t \sqrt{v(s)}\dif W_j(s)+\sum_{n=1}^{N_j(t)}Y_{j,n}\Bigg\},$$ for $0<t\leq T$. Assumption \ref {asp-PutCall-VarianceParameterAssumptions} and the non-explosion assumption on the point processes imply that the integrals and summation that appear above are well-defined. It also follows that $S_j(t)>0$ $\prob$-a.s. for all $t\in[0,T]$, and hence either asset can be used as a num\'eraire. 

Instead of the money market account, we take $\{S_2(t)e^{q_2 t}\}$, the second asset yield process, as the num\'eraire and define the probability measure $\hat{\Q}$, equivalent to $\prob$, such that the first asset yield process and the money market account, when discounted by $S_2(t)e^{q_2 t}$, are martingales under $\hat{\Q}$. With the second asset yield process as the num\'eraire, the \emph{discounted price} of any other asset with price process $\{X(t)\}$ is defined by $\tilde{X}(t) = X(t)/(S_2(t)e^{q_2 t})$. In the absence of arbitrage opportunities, the appropriate discounting factor for the period $[t,t']$ for any $t'\in[t,T]$ is given by $DF(t,t') = S_2(t)e^{q_2 t}/(S_2(t')e^{q_2 t'})$.

Next, we discuss the construction of the equivalent probability measure $\hat{\Q}$. The following standard proposition specifies the form of the Radon-Nikod\'ym derivative $\frac{\dif\hat{\Q}}{\dif\prob}$.

\begin{prop}
\label{prop-PutCall-ChangeofMeasure}
Suppose $\bm{\theta}(t) = \left(\psi_{1}(t),\psi_{2}(t),\zeta(t)\right)^\top$ is a vector of $\calF_t$-adapted processes and let $\gamma_1,\gamma_2,\nu_1,\nu_2$ be constants. Define the process $\{L_t\}$ by 
\begin{align}
\begin{split}
\label{eqn-PutCall-RNDerivative}
L(t)	& = \exp\left\{-\int_0^t\left(\bm{\Sigma}^{-1}\bm{\theta}(s)\right)^\top\dif\vm{B}(s)-\frac{1}{2}\int_0^t\bm{\theta}(s)^\top\bm{\Sigma}^{-1}\bm{\theta}(s)\dif s\right\}\\
		& \qquad \times\exp\left\{\sum_{n=1}^{N_{1}(t)}(\gamma_1 Y_{1,n}+\nu_1)-\lambda_1 t\left(e^{\nu_1}\E_\prob(e^{\gamma_1 Y_{1}})-1\right)\right\}\\
		& \qquad \times\exp\left\{\sum_{n=1}^{N_{2}(t)}(\gamma_2 Y_{2,n}+\nu_2)-\lambda_2 t\left(e^{\nu_2}\E_\prob(e^{\gamma_2 Y_{2}})-1\right)\right\}
\end{split}
\end{align}
and suppose that $\{L(t)\}$ is a strictly positive $\prob$-martingale such that $\E_\prob[L(t)]=1$ for all $t\in[0,T]$. Then $L(T)$ is the Radon-Nikod\'ym derivative of some probability measure $\hat{\Q}$ equivalent to $\prob$ and the following hold:
\begin{enumerate}
	\item Under $\hat{\Q}$, the vector process $\vm{B}(t)$ has drift $-\bm{\theta}(t)$;
	\item The Poisson process $N_j(t)$ has a new intensity $\tilde{\lambda}_j = \lambda_j e^{\nu_j}\E_\prob[e^{\gamma_j Y_j}]$, $j=1,2$	under $\hat{\Q}$; and
	\item The moment generating function of jump sizes random variable $Y_j$ under $\hat{\Q}$ is given by $M_{\hat{\Q},Y_j}(u) = M_{\prob,Y_j}(u+\gamma_j)/M_{\prob,Y_j}(\gamma_j)$, $j=1,2.$
\end{enumerate}
\end{prop}

\begin{proof}
See e.g. \citet[Theorem 2.4]{Runggaldier-2003} and \citet[Theorem 1]{CheangTeh-2014}.
\end{proof}

The Radon-Nikod\'ym derivative $L(T) = \frac{\dif\hat{\Q}}{\dif\prob}$ can be used to characterize any probability measure $\hat{\Q}$ equivalent to $\prob$ as parameterized by the vector process $\{\bm{\theta}(t)\}$ and the constants $\gamma_1,\gamma_2,\nu_1,\nu_2$. We assume that $\gamma_1,\gamma_2,\nu_1,\nu_2$ are constant to preserve the time-homogeneity of the intensity and the jump size distribution. As the market under the SVJD is generally incomplete, one can construct multiple equivalent martingale measures consistent with the no-arbitrage assumption.

We now specify the parameters of $L(T)$ so that $\hat{\Q}$ becomes an equivalent martingale measure corresponding to the num\'eraire $\{S_2(t)e^{q_2 t}\}$. Let $\{\tilde{S}(t)\}$ and $\{\tilde{M}(t)\}$, where $\tilde{S}(t) = S_1(t)e^{q_1 t}/(S_2(t)e^{q_2 t})$ and $\tilde{M}(t) = e^{rt}/(S_2(t)e^{q_2 t}),$ be the first asset yield process and the money market account when discounted using the second stock's yield process. In particular, we will refer to $\{\tilde{S}(t)\}$ as the \emph{asset yield ratio process}. If we choose $\{\psi_1(t)\}$, $\{\psi_2(t)\}$, and $\{\zeta(t)\}$ as
\begin{align}
\label{eqn-PutCall-RiskPremium-1}
\psi_1(t) & = \frac{\mu_1+q_1-r-\rho_w\sigma_1\sigma_1 v(t)-\lambda_1\kappa_1+\tilde{\lambda}_1\tilde{\kappa}_1}{\sigma_1\sqrt{v(t)}}\\
\label{eqn-PutCall-RiskPremium-2}
\psi_2(t) & = \frac{\mu_2+q_2-r-\sigma_2^2 v(t)-\lambda_2\kappa_2-\tilde{\lambda}_2\tilde{\kappa}_2^-}{\sigma_2\sqrt{v(t)}}\\
\label{eqn-PutCall-MarketPriceofVolRisk}
\zeta(t) & = \frac{\Lambda}{\omega}\sqrt{v(t)} \qquad \text{for some constant $\Lambda\geq 0$},
\end{align} 
where $\tilde{\kappa}_1 = \E_{\hat{\Q}}[e^{Y_1}-1]$ and $\tilde{\kappa}_2^- = \E_{\hat{\Q}}[e^{-Y_2}-1]$, then $\{\tilde{S}(t)\}$ and $\{\tilde{M}(t)\}$ are $\hat{\Q}$-martingales on $[0,T]$.\footnote{This assertion can be proved using It\^o's Lemma on $\tilde{S}(t)$ and $\tilde{M}(t)$ and eliminating the resulting drift term as required by the martingale representation for jump-diffusion processes \citep[see][Theorem 2.3]{Runggaldier-2003}.}

With this choice of parameters for $\hat{\Q}$, the dynamics of the instantaneous variance becomes
\begin{equation}
\label{eqn-PutCall-VarianceSDE-Q}
\dif v(t) = \left[\xi\eta-(\xi+\Lambda)v(t)\right]\dif t+\omega\sqrt{v(t)}\dif\bar{Z}(t).
\end{equation}
where $\{\bar{Z}(t)\}$ is a $\hat{\Q}$-Wiener process. The choice of $\zeta(t)$ preserves the structure of the instantaneous variance as a square-root process. Assumption \ref{asp-PutCall-VarianceParameterAssumptions} ensures that this process is strictly positive and finite $\hat{\Q}$-a.s.

Under $\hat{\Q}$, $\tilde{S}(t)$ satisfies the equation
\begin{align}
\begin{split}
\label{eqn-PutCall-YieldRatioSDE-Q}
\dif\tilde{S}(t)
	& = -\tilde{S}(t)\left(\tilde{\lambda}_1\tilde{\kappa}_1+\tilde{\lambda}_2\tilde{\kappa}_2^-\right)\dif t + \sigma\sqrt{v(t)}\tilde{S}(t)\dif\bar{W}(t)\\
	& \qquad + \int_{\mathbb{R}} \left(e^{y_1}-1\right)\tilde{S}(t) p(\dif y_1,\dif t) + \int_{\mathbb{R}}\left(e^{-y_2}-1\right)\tilde{S}(t)p(\dif y_2,\dif t).
\end{split}
\end{align}
where we define $\sigma\dif\bar{W}(t) \equiv \sigma_1\dif\bar{W}_1(t)-\sigma_2\dif\bar{W}_2(t)$ with standard $\hat{\Q}$-Wiener processes $\{\bar{W}_1(t)\}$ and $\{\bar{W}_2(t)\}$ and $\sigma^2 = \sigma_1^2+\sigma_2^2-2\rho_w\sigma_1\sigma_2$.\footnote{In view of Proposition \ref{prop-PutCall-ChangeofMeasure}, we note that $\dif\bar{W}_j(t) = \psi_j(t)\dif t+\dif W_j(t)$.} This equation admits a solution $\tilde{S}(t)$ given by
\begin{align*}
\tilde{S}(t)
	& = \tilde{S}(0)\exp\Bigg\{-(\tilde{\lambda}_1\tilde{\kappa}_1+\tilde{\lambda}_2\tilde{\kappa}_2^-)t - \frac{1}{2}\sigma^2\int_0^t v(s)\dif s+\sigma\int_0^t\sqrt{v(s)}\dif\bar{W}(s)\\
	& \qquad \qquad \qquad +\sum_{m=1}^{N_1(t)}Y_{1,m}-\sum_{n=1}^{N_2(t)}Y_{2,n}\Bigg\}
\end{align*}

Lastly, we note that the instantaneous correlation between the $\hat{\Q}$-Brownian motions $\{\bar{W}(t)\}$ and $\{\bar{Z}(t)\}$ is given by $$\E_{\hat{\Q}}\left[\dif\bar{W}(t)\dif\bar{Z}(t)\right] = \frac{1}{\sigma}\left(\sigma_1 \rho_1-\sigma_2\rho_2\right)\dif t = \frac{\sigma_1 \rho_1-\sigma_2\rho_2}{\sqrt{\sigma_1^2+\sigma_2^2-2\rho_w\sigma_1\sigma_2}}\dif t.$$ 

\section{The Exchange Option Pricing IPDE}
\sectionmark{Exchange Option IPDE}
\label{sec-PutCall-ExchangeOptionIPDE}

Now we derive the integro-partial differential equation (IPDE) for the price of an exchange option written on $S_1$ and $S_2$. Denote by $C(t,S_1(t),S_2(t),v(t))$ the price of a European exchange option whose terminal payoff is given by $C\left(T,S_1(T),S_2(T),v(T)\right) = \left(S_1(T)-S_2(T)\right)^+,$ where $x^+\equiv\max\{x,0\}$. A rearrangement of terms expresses the discounted terminal payoff as
\begin{equation*}
\frac{C\left(T,S_1(T),S_2(T),v(T)\right)}{S_2(T)e^{q_2 T}} = e^{-q_1 T}\left(\tilde{S}(T)-e^{(q_1-q_2)T}\right)^+.
\end{equation*}
Let $\tilde{C}(t,S_1(t),S_2(t),v(t)) \equiv C\left(t,S_1(t),S_2(t),v(t)\right)/\left(S_2(t)e^{q_2 t}\right)$ denote the discounted European exchange option price. Then, assuming that no arbitrage opportunities exist, $\tilde{C}(t,S_1(t),S_2(t),v(t))$ is given by
\begin{align}
\begin{split}
\label{eqn-PutCall-DiscountedExcOpPrice}
\tilde{C}\left(t,S_1(t),S_2(t),v(t)\right)
	& = \E_{\hat{\Q}}\left[\left.\tilde{C}\left(T,S_1(T),S_2(T),v(T)\right)\right|\calF_t\right]\\
	& = e^{-q_1 T}\E_{\hat{\Q}}\left[\left.\left(\tilde{S}(T)-e^{(q_1-q_2)T}\right)^+\right|\calF_t\right].
\end{split}
\end{align}
In other words, the price at any time $t<T$ of the European exchange option measured in units of the second asset yield process is the expected value, under the probability measure $\hat{\Q}$, of the terminal payoff measured in units of the second asset yield process \citep{Geman-1995}. From the last equation, we also note that the terminal payoff is variable only in the asset yield ratio $\tilde{S}(t)$. Thus, we assume that the discounted European exchange option price is represented by the process $\tilde{V}(t,\tilde{S}(t),v(t)) \equiv \tilde{C}\left(t,S_1(t),S_2(t),v(t)\right),$ and so
\begin{equation}
\label{eqn-PutCall-DiscountedExcOpPrice2}
\tilde{V}(t,\tilde{S}(t),v(t)) = e^{-q_1 T}\E_{\hat{\Q}}\left[\left.\left(\tilde{S}(T)-e^{(q_1-q_2)T}\right)^+\right|\calF_t\right].
\end{equation}

At this point, we have shown that, by taking the second stock's yield process as the num\'eraire asset, \emph{the exchange option pricing problem is equivalent to pricing a European call option on the asset yield price ratio $\tilde{S}(t)$ with maturity date $T$ and strike price $e^{(q_1-q_2)T}$}. In the succeeding analysis, we shall take advantage of this simplification and employ techniques in pricing European call options under stochastic volatility and jump-diffusion dynamics \citep[e.g.][]{Bates-1996, CheangChiarellaZiogas-2013}.

\begin{rem}
If we choose the first asset yield process $\{S_1(t)e^{q_1 t}\}$ as the num\'eraire, then the exchange option pricing problem simplifies to the valuation of a \emph{put option} written on the asset yield ratio $(S_{2}(t)e^{q_2 t})/(S_1(t)e^{q_1 t})$.
\end{rem}

The following technical assumption is required to implement It\^o's formula for jump-diffusion processes.

\begin{asp}
\label{asp-PutCall-Differentiability}
For $t\in[0,T]$, $\tilde{V}(t,\tilde{s},v)$ is (at least) twice-differentiable in $\tilde{s}$ and $v$ and differentiable in $t$ with continuous partial derivatives.
\end{asp}

In the following proposition, we derive the IPDE that characterizes the discounted European exchange option price.

\begin{prop}
\label{prop-PutCall-DiscountedEuExcOpPrice}
The price at time $t\in[0,T)$ of the European exchange option is given by
\begin{equation}
C(t,S_1(t),S_2(t),v(t)) = S_2(t)e^{q_2 t}\tilde{V}(t,\tilde{S}(t),v(t)),
\end{equation}
where $\tilde{V}$, satisfying Assumption \ref{asp-PutCall-Differentiability}, is the solution of the terminal value problem
\begin{align}
\label{eqn-PutCall-IPDE-Vtilde}
0 & = \pder[\tilde{V}]{t}+\calL_{\tilde{s},v}\left[\tilde{V}(t,\tilde{S}(t),v(t))\right], \qquad (t,\tilde{S}(t),\tilde{v}(t))\in[0,T]\times\mathbb{R}_+^2\\
\label{eqn-PutCall-IPDE-Vtilde-TerminalCondition}
\tilde{V}(T) & = e^{-q_1 T}\left(\tilde{S}(T)-e^{(q_1-q_2)T}\right)^+,
\end{align}
with $\mathbb{R}_+^2 = (0,\infty)\times(0,\infty)$ and the IPDE operator $\calL_{\tilde{s},v}$ defined as 
\begin{align}
\begin{split}
\label{eqn-PutCall-IPDEOperator}
\calL_{\tilde{s},v}\left[\tilde{V}(t,\tilde{S},v)\right] 
	& = -\tilde{S}\left(\tilde{\lambda}_1\tilde{\kappa}_1+\tilde{\lambda}_2\tilde{\kappa}_2^-\right)\pder[\tilde{V}]{\tilde{s}}+\left[\xi\eta-(\xi+\Lambda)v\right]\pder[\tilde{V}]{v}\\
	& \qquad + \frac{1}{2}\sigma^2 v\tilde{S}^2\pder[^2\tilde{V}]{\tilde{s}^2}+\frac{1}{2}\omega^2 v\pder[^2\tilde{V}]{v^2}+\omega(\sigma_1 \rho_1-\sigma_2\rho_2)v\tilde{S}\pder[^2\tilde{V}]{\tilde{s}\partial v}\\
	& \qquad + \tilde{\lambda_1}\E_{\hat{\Q}}^{Y_1}\left[\tilde{V}\left(t,\tilde{S}e^{Y_1},v\right)-\tilde{V}(t,\tilde{S},v)\right]\\
	& \qquad + \tilde{\lambda}_2\E_{\hat{\Q}}^{Y_2}\left[\tilde{V}\left(t,\tilde{S}e^{-Y_2},v\right)-\tilde{V}(t,\tilde{S},v)\right],
\end{split}
\end{align}
where $\E_{\hat{\Q}}^{Y_i}$ is the expectation with respect to the r.v. $Y_i$ ($i=1,2$) under the measure $\hat{\Q}$. Note that all partial derivatives are evaluated at $(t,\tilde{S}(t),v(t))$.
\end{prop}

\begin{proof}
The tower property for conditional expectations imply that $\{\tilde{V}(t)\}$ is a $\hat{\Q}$-martingale, with integrability guaranteed by Assumption \ref{asp-PutCall-Differentiability}. With equations \eqref{eqn-PutCall-VarianceSDE-Q} and \eqref{eqn-PutCall-YieldRatioSDE-Q} in mind, an application of It\^o's formula shows that $\tilde{V}(t)$ satisfies the stochastic differential equation
\begin{align}
\begin{split}
\label{eqn-PutCall-SDE-V-tilde}
\dif\tilde{V}(t)
	& = \left\{\pder[\tilde{V}]{t}+\calL_{\tilde{s},v}\left[\tilde{V}(t,\tilde{S}(t),v(t))\right]\right\}\dif t\\
	& \qquad + \sigma\sqrt{v(t)}\tilde{S}(t)\pder[\tilde{V}]{\tilde{s}}\dif\bar{W}(t)+\omega\sqrt{v(t)}\pder[\tilde{V}]{v}\dif\bar{Z}(t)\\
	& \qquad + \int_{\mathbb{R}}\left[\tilde{V}\left(t,\tilde{S}(t)e^{y_1},v(t)\right)-\tilde{V}(t,\tilde{S}(t),v(t))\right]q(\dif y_1,\dif t)\\
	& \qquad + \int_{\mathbb{R}}\left[\tilde{V}\left(t,\tilde{S}(t)e^{-y_2},v(t)\right)-\tilde{V}(t,\tilde{S}(t),v(t))\right]q(\dif y_2,\dif t),
\end{split}
\end{align}
where $\calL_{\tilde{s},v}$ is the IPDE operator defined by equation \eqref{eqn-PutCall-IPDEOperator}. Since $\{\tilde{V}(t)\}$ is a $\hat{\Q}$-martingale, the drift must be equal to zero, giving us the equation \eqref{eqn-PutCall-IPDE-Vtilde}. Terminal condition \eqref{eqn-PutCall-IPDE-Vtilde-TerminalCondition} follows from the discussion at the start of this section.
\end{proof}

Let $C^A(t,S_1(t),S_2(t),v(t))$ be the price at time $t$ of an American exchange option written on $S_1$ and $S_2$. After a rearrangement of terms, standard theory on American option pricing \citep[see e.g.][]{Myneni-1992} dictates that the discounted American exchange option price $\tilde{V}^A(t,\tilde{S}(t),v(t))$ is given by
\begin{align}
\begin{split}
\tilde{V}^A(t,\tilde{S}(t),v(t)) 
	& \equiv \frac{C^A(t,S_1(t),S_2(t),v(t)}{S_2(t)e^{q_2 t}}\\
	& = \esssup_{u\in[t,T]} e^{-q_1 u}\E_{\hat{\Q}}\left[\left.\left(\tilde{S}(u)-e^{(q_1-q_2)u}\right)^+\right|\calF_t\right],
\end{split}
\end{align}
where the supremum is taken over all $\hat{\Q}$-stopping times $u\in[t,T]$. \emph{From here, we see that the change of num\'eraire reduces the problem to pricing an American call option on the asset yield price ratio $\tilde{S}(t)$ with maturity date $T$ and strike price $e^{(q_1-q_2)T}$,} similar to our observation for the European exchange option. The price of the American exchange option also hedges against the exchange option payoff in the sense that
\begin{align*}
\tilde{V}^A(t,\tilde{S}(t),v(t)) & \geq e^{-q_1 t}\left(\tilde{S}(t)-e^{(q_1-q_2)t}\right)^+ \qquad \forall t\in [0,T)\\
\tilde{V}^A(T,\tilde{S}(T),v(T)) & = e^{-q_1 T}\left(\tilde{S}(T)-e^{(q_1-q_2)T}\right)^+.
\end{align*}

Before prescribing additional boundary conditions to IPDE \eqref{eqn-PutCall-IPDE-Vtilde} for the American exchange option, we first define the continuation and stopping regions, denoted by $\calC$ and $\calS$, respectively, that divide the domain $[0,T]\times\mathbb{R}_+^2$ of IPDE \eqref{eqn-PutCall-IPDE-Vtilde}. These regions are given by
\begin{align}
\begin{split}
\label{eqn-PutCall-StoppingContinuationRegions}
\calS & = \left\{(t,\tilde{S},v)\in[0,T]\times\mathbb{R}_+^2: \tilde{V}^A(t,\tilde{S},v) = e^{-q_1 t}\left(\tilde{S}-e^{(q_1-q_2)t}\right)^+\right\}\\
\calC & = \left\{(t,\tilde{S},v)\in[0,T]\times\mathbb{R}_+^2: \tilde{V}^A(t,\tilde{S},v) > e^{-q_1 t}\left(\tilde{S}-e^{(q_1-q_2)t}\right)^+\right\}.
\end{split}
\end{align}
Denote by $\calS(t)$ and $\calC(t)$ the stopping and continuation regions at a fixed $t\in[0,T]$.

From \citet{BroadieDetemple-1997}, there exists a critical stock price ratio $B(t,v)\geq 1$, dependent on the current variance level \citep{Touzi-1999}, such that the stopping and continuation regions can be written as
\begin{align}
\begin{split}
\label{eqn-PutCall-StoppingContinuationRegions2}
\calS & = \left\{(t,\tilde{S},v)\in[0,T]\times\mathbb{R}_+^2: \tilde{S}\geq B(t,v)e^{(q_1-q_2)t}\right\}\\
\calC & = \left\{(t,\tilde{S},v)\in[0,T]\times\mathbb{R}_+^2: \tilde{S}< B(t,v)e^{(q_1-q_2)t}\right\}.
\end{split}
\end{align}
The line $s_1 = B(t,v)s_2$ on the $s_1 s_2$-plane is known as the early exercise boundary.\footnote{\citet{Mishura-2009} analyzed, in further detail, the properties of the exercise region of the finite-maturity American exchange option in a pure diffusion setting. In the same setting, \citet{Villeneuve-1999} established the nonemptiness of exercise regions of American rainbow options, which include spread and exchange options as special cases.} For a fixed $t\in[0,T]$ and $v\in(0,\infty)$, the early exercise boundary and the continuation and stopping regions are illustrated in Figure \ref{fig-EarlyExerciseBoundaryOptionPrice}. It is known that in the continuation region the American exchange option behaves like its live European counterpart, and so $\tilde{V}^A$ satisfies IPDE \eqref{eqn-PutCall-IPDE-Vtilde} for $(t,\tilde{S},v)\in\calC$.

\begin{figure}
\centering
\includegraphics[width = 5in]{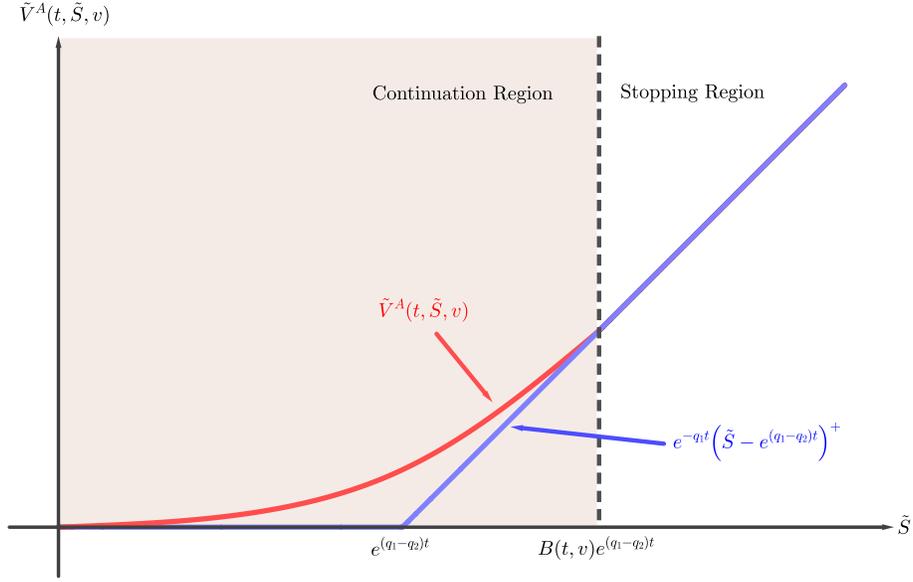}
\caption{The early exercise boundary and the continuation and stopping regions for the transformed American exchange option.}
\label{fig-EarlyExerciseBoundaryOptionPrice}
\end{figure}

We require value-matching and smooth-pasting conditions on IPDE \eqref{eqn-PutCall-IPDE-Vtilde} to enforce the no-arbitrage assumption and to ensure that the discounted exchange option price $\tilde{V}^A$ and its partial derivative $\partial\tilde{V}^A/\partial\tilde{s}$ are both continuous across the early exercise boundary $A(t,v)\equiv B(t,v)e^{(q_1-q_2)t}$. Specifically, the required value-matching condition is
\begin{equation}
\label{eqn-PutCall-ValueMatchingCondition-Vtilde}
\tilde{V}^A(t,A(t,v),v(t)) = e^{-q_1 t}\left(A(t,v)-e^{(q_1-q_2)t}\right),
\end{equation}
and the smooth-pasting conditions are
\begin{align}
\begin{split}
\label{eqn-PutCall-SmoothPastingCondition-Vtilde}
\lim_{\tilde{S}\to A(t,v)} \pder[\tilde{V}^A]{\tilde{s}}(t,\tilde{S}(t),v(t)) & = e^{-q_1 t}\\
\lim_{\tilde{S}\to A(t,v)} \pder[\tilde{V}^A]{v}(t,\tilde{S}(t),v(t)) & = 0\\
\lim_{\tilde{S}\to A(t,v)} \pder[\tilde{V}^A]{t}(t,\tilde{S}(t),v(t)) & = -q_1 e^{-q_1 t}\tilde{S}(t)+q_2 e^{-q_2 t}.
\end{split}
\end{align}
Therefore, the discounted American exchange option price is a solution to IPDE \eqref{eqn-PutCall-IPDE-Vtilde} over the domain $0\leq t\leq T$, $0<\tilde{S}<A(t,v)$, $0<v<\infty$. The IPDE has terminal and boundary conditions
\begin{align}
\begin{split}
\label{eqn-PutCall-BoundaryConditions-Vtilde}
\tilde{V}(T,\tilde{S}(T),v(T)) & = e^{-q_1 T}\left(\tilde{S}(T)-e^{(q_1-q_2)T}\right)^+\\
\tilde{V}(t,0,v(t)) & = 0,
\end{split}
\end{align}
value-matching condition \eqref{eqn-PutCall-ValueMatchingCondition-Vtilde} and smooth-pasting condition \eqref{eqn-PutCall-SmoothPastingCondition-Vtilde}.

\section{An Early Exercise Representation}
\sectionmark{Early Exercise Representation}
\label{sec-PutCall-EER}

In this section, we show that $\tilde{V}^A(t,\tilde{S}(t),v(t))$ can be decomposed into the sum of the discounted European exchange option price $\tilde{V}(t,\tilde{S}(t),v(t))$ and an early exercise premium.

\begin{prop}
\label{prop-PutCall-EarlyExerciseRepresentation}
Suppose Assumption \ref{asp-PutCall-Differentiability} also holds for $\tilde{V}^A(t,\tilde{S},v)$. Assume further that the smooth pasting conditions \eqref{eqn-PutCall-SmoothPastingCondition-Vtilde} across the early exercise boundary hold. 
Then $\tilde{V}^A(t,\tilde{S}(t),v(t))$ can be expressed as
\begin{equation}
\label{eqn-PutCall-EarlyExerciseRepresentation}
\tilde{V}^A(t,\tilde{S}(t),v(t)) = \tilde{V}(t,\tilde{S}(t),v(t)) + \tilde{V}^P(t,\tilde{S}(t),v(t)),
\end{equation}
where $\tilde{V}(t,\tilde{S}(t),v(t))$ is the discounted European exchange option price given by equation \eqref{eqn-PutCall-DiscountedExcOpPrice2} and $\tilde{V}^P(t,\tilde{S}(t),v(t))$ is the early exercise premium given by 
\begin{equation}
\label{eqn-PutCall-EarlyExercisePremium}
\tilde{V}^P(t,\tilde{S}(t),v(t)) = -\E_{\hat{\Q}}\left[\left.\int_{t}^T\left\{\pder[\tilde{V}^A]{t}+\calL_{\tilde{s},v}\left[\tilde{V}^A(s,\tilde{S}(s),v(s))\right]\right\}\dif s\right|\calF_t\right].
\end{equation}
Note that all partial derivatives in $\partial\tilde{V}^A/\partial t+\calL_{\tilde{s},v}[\tilde{V}^A(s,\tilde{S}(s),v(s))]$ are all evaluated at $(s,\tilde{S}(s),v(s))$.
\end{prop}

\begin{proof}
Given that Assumption \ref{asp-PutCall-Differentiability} also applies to $\tilde{V}^A(t,\tilde{S}(t),v(t))$ and equipped with the smooth-pasting conditions discussed above, an application of It\^o's formula verifies that $\tilde{V}^A(t)\equiv\tilde{V}^A(t,\tilde{S}(t),v(t))$ satisfies equation \eqref{eqn-PutCall-SDE-V-tilde}. Integrating over $[t_0,t]$, where $0\leq t_0\leq t\leq T$, we find that
\begin{align*}
\tilde{V}^A(t)
	& = \tilde{V}^A(t_0) + \int_{t_0}^t\left\{\pder[\tilde{V}^A]{t}+\calL_{\tilde{s},v}\left[\tilde{V}^A(s,\tilde{S}(s),v(s))\right]\right\}\dif s\\
	& \qquad + \int_{t_0}^t \sigma\sqrt{v(s)}\tilde{S}(s)\pder[\tilde{V}^A]{\tilde{s}}\dif\bar{W}(s)+\int_{t_0}^t \omega\sqrt{v(s)}\pder[\tilde{V}^A]{v}\dif\bar{Z}(s)\\
	& \qquad + \int_{t_0}^t \int_{\mathbb{R}}\left[\tilde{V}^A\left(s,\tilde{S}(s)e^{y_1},v(s)\right)-\tilde{V}^A(s,\tilde{S}(s),v(s))\right]q(\dif y_1,\dif s)\\
	& \qquad + \int_{t_0}^t \int_{\mathbb{R}}\left[\tilde{V}^A\left(s,\tilde{S}(s)e^{-y_2},v(s)\right)-\tilde{V}^A(s,\tilde{S}(s),v(s))\right]q(\dif y_2,\dif s).
\end{align*}
Next, we take the $\hat{\Q}$-expectation of the above equation conditional on $\calF_{t_0}$. We note that the integrals with respect to the Wiener processes and the compensated counting measures are all independent of $\calF_{t_0}$ and that the unconditional expectation of the integrals with respect to the Wiener processes is zero. These observations, combined with the martingale representation theorem for marked point processes \citep[see][]{Bremaud-1981}, imply that
\begin{equation*}
\E_{\hat{\Q}}\left[\left.\tilde{V}^A(t)\right|\calF_{t_0}\right] = \tilde{V}^A(t_0)+\E_{\hat{\Q}}\left[\left.\int_{t_0}^t\left\{\pder[\tilde{V}^A]{t}+\calL_{\tilde{s},v}\left[\tilde{V}^A(s,\tilde{S}(s),v(s))\right]\right\}\dif s\right|\calF_{t_0}\right].
\end{equation*}
Set $t=T$ and $t_0 = t$. Terminal condition \eqref{eqn-PutCall-BoundaryConditions-Vtilde} and equation \eqref{eqn-PutCall-DiscountedExcOpPrice2} imply that the left-hand side of this equation is equal to $\tilde{V}(t,\tilde{S}(t),v(t))$, the price of the European exchange option. Rearranging yields the result stated in the proposition.
\end{proof}

For a fixed $t\in[0,T]$, let $\calA(t)$ be the event that $\tilde{S}(t)$ and $v(t)$ are in the stopping region $\calS(t)$; that is, $\calA(t) \equiv \{(\tilde{S}(t),v(t))\in\calS(t)\}.$ The complement event $\calA^c(t)$ denotes the event that $\tilde{S}(t)$ and $v(t)$ are in the continuation region $\calC(t)$.

We now seek to evaluate the expectation defining the early exercise premium in equation \eqref{eqn-PutCall-EarlyExercisePremium}. This is discussed in the next proposition.

\begin{prop}
\label{prop-PutCall-EarlyExercisePremium}
The early exercise premium $\tilde{V}^P(t,\tilde{S}(t),v(t))$ is given by
\begin{align}
\begin{split}
\label{eqn-PutCall-EarlyExercisePremium2}
& \tilde{V}^P(t,\tilde{S}(t),v(t))\\ 
& \qquad = \E_{\hat{\Q}}\left[\left.\int_t^T\left(q_1 e^{-q_1 s}\tilde{S}(s)-q_2 e^{-q_2 s}\right)\vm{1}(\calA(s))\dif s\right|\calF_t\right]\\
& \qquad \qquad -\tilde{\lambda}_1 \E_{\hat{\Q}}\left[\int_t^T \E_{\hat{\Q}}^{Y_1}\left[\left(\tilde{V}^A\left(s,\tilde{S}(s)e^{Y_1},v(s)\right)\right.\right.\right.\\
& \hspace{100pt} -\left.\left.\left.\left.\left(e^{-q_1 s}\tilde{S}(s)e^{Y_1}-e^{-q_2 s}\right)\right)\vm{1}(\calA_1(s))\right]\dif s\right|\calF_t\right]\\
& \qquad \qquad -\tilde{\lambda}_2 \E_{\hat{\Q}}\left[\int_t^T \E_{\hat{\Q}}^{Y_2}\left[\left(\tilde{V}^A\left(s,\tilde{S}(s)e^{-Y_2},v(s)\right)\right.\right.\right.\\
& \hspace{100pt} -\left.\left.\left.\left.\left(e^{-q_1 s}\tilde{S}(s)e^{-Y_2}-e^{-q_2 s}\right)\right)\vm{1}(\calA_2(s))\right]\dif s\right|\calF_t\right],
\end{split}
\end{align}
where $\vm{1}(\cdot)$ is the indicator function, $\calA(s)$ is the event $\{(\tilde{S}(t),v(t))\in\calS(t)\}$ and
\begin{align*}
\calA_1(s)	& = \left\{B(s,v)e^{(q_1-q_2)s}\leq \tilde{S}(s) < B(s,v)e^{(q_1-q_2)s}e^{-Y_1}\right\}\\
\calA_2(t)	& = \left\{B(s,v)e^{(q_1-q_2)s}\leq \tilde{S}(s) < B(s,v)e^{(q_1-q_2)s}e^{Y_2}\right\}.
\end{align*}
\end{prop}

\begin{proof}
Note first that for any time $s\in[t,T]$,
\begin{align*}
\pder[\tilde{V}^A]{t}+\calL_{\tilde{s},v}\left[\tilde{V}^A(s,\tilde{S}(s),v(s))\right]
	& = \left(\pder[\tilde{V}^A]{t}+\calL_{\tilde{s},v}\left[\tilde{V}^A(s,\tilde{S}(s),v(s))\right]\right)\vm{1}(\calA(t))\\
	& \qquad +\left(\pder[\tilde{V}^A]{t}+\calL_{\tilde{s},v}\left[\tilde{V}^A(s,\tilde{S}(s),v(s))\right]\right)\vm{1}(\calA^c(t))\\
	& = \left(\pder[\tilde{V}^A]{t}+\calL_{\tilde{s},v}\left[\tilde{V}^A(s,\tilde{S}(s),v(s))\right]\right)\vm{1}(\calA(t)),
\end{align*}
since in the continuation region, the American exchange option behaves like its live European counterpart and the integro-partial differential terms vanish (see equation \eqref{eqn-PutCall-IPDE-Vtilde}). In the stopping region (i.e. if $\vm{1}(\calA(s))=1$), we note that $\tilde{V}^A(s,\tilde{S}(s),v(s)) = e^{-q_1 s}\tilde{S}(s)-e^{-q_2 s}$ (see equation \eqref{eqn-PutCall-StoppingContinuationRegions}). Applying the integro-partial differential operators, recalling the definition of $\tilde{\kappa}_1$ and $\tilde{\kappa}_2^-$, and rearranging the terms, we find that
\begin{align}
\begin{split}
\label{eqn-PutCall-IPDE-StoppingRegion}
& \left(\pder[\tilde{V}^A]{t}+\calL_{\tilde{s},v}\left[\tilde{V}^A(s,\tilde{S}(s),v(s))\right]\right)\vm{1}(\calA(s))\\
& \qquad = \left[-q_1 e^{-q_1 s}\tilde{S}(s)+q_2 e^{-q_2 s}\right]\vm{1}(\calA(s))\\
& \qquad \qquad + \tilde{\lambda}_1\E_{\hat{\Q}}^{Y_1}\left[\tilde{V}^A\left(s,\tilde{S}(s)e^{Y_1},v(s)\right)-\left(e^{-q_1 s}\tilde{S}(s)e^{Y_1}-e^{-q_2 s}\right)\right]\vm{1}(\calA(s))\\
& \qquad \qquad + \tilde{\lambda}_2\E_{\hat{\Q}}^{Y_2}\left[\tilde{V}^A\left(s,\tilde{S}(s)e^{-Y_2},v(s)\right)-\left(e^{-q_1 s}\tilde{S}(s)e^{-Y_2}-e^{-q_2 s}\right)\right]\vm{1}(\calA(s)).
\end{split}
\end{align}
Observe that the expectations above contain option prices determined \emph{after the jump} in $\tilde{S}(s)$ occurring at time $s$. 

As stated in the proposition, define $\calA_1(s)$ and $\calA_2(s)$ as the events in which the asset price ratio is initially in the stopping region but is sent back into the continuation region after a jump by a factor $e^{Y_1}$ and $e^{-Y_2}$, respectively, at time $s$. From the stopping and continuation criteria in equation \eqref{eqn-PutCall-StoppingContinuationRegions}, we note that 
\begin{equation}
\label{eqn-PutCall-StoppingContinuationCriterion-Y1}
\tilde{V}^A\left(s,\tilde{S}(s)e^{Y_1},v(s)\right) \geq e^{-q_1 s}\tilde{S}(s)e^{Y_1}-e^{-q_2 s},
\end{equation} 
with the strict inequality occurring when $\calA_1(s)$ is true, and 
\begin{equation}
\label{eqn-PutCall-StoppingContinuationCriterion-Y2}
\tilde{V}^A\left(s,\tilde{S}(s)e^{-Y_2},v(s)\right) \geq e^{-q_1 s}\tilde{S}(s)e^{-Y_2}-e^{-q_2 s},
\end{equation} 
with the strict inequality occurring when $\calA_2(t)$ is true. Therefore, we have
\begin{align*}
& \left(\pder[\tilde{V}^A]{t}+\calL_{\tilde{s},v}\left[\tilde{V}^A(s,\tilde{S}(s),v(s))\right]\right)\vm{1}(\calA(s))\\
& \qquad = -\left[q_1 e^{-q_1 s}\tilde{S}(s)-q_2 e^{-q_2 s}\right]\vm{1}(\calA(s))\\
& \qquad \qquad + \tilde{\lambda}_1\E_{\hat{\Q}}^{Y_1}\left[\left(\tilde{V}^A\left(s,\tilde{S}(s)e^{Y_1},v(s)\right)-\left(e^{-q_1 s}\tilde{S}(s)e^{Y_1}-e^{-q_2 s}\right)\right)\vm{1}(\calA_1(s))\right]\\
& \qquad \qquad + \tilde{\lambda}_2\E_{\hat{\Q}}^{Y_2}\left[\left(\tilde{V}^A\left(s,\tilde{S}(s)e^{-Y_2},v(s)\right)-\left(e^{-q_1 s}\tilde{S}(s)e^{-Y_2}-e^{-q_2 s}\right)\right)\vm{1}(\calA_2(s))\right].
\end{align*}
Using the above expression in the early exercise premium given in equation \eqref{eqn-PutCall-EarlyExercisePremium}, we obtain \eqref{eqn-PutCall-EarlyExercisePremium2}.
\end{proof}

\begin{rem}
If $G_1(y)$ and $G_2(y)$ are the probability density functions (pdfs) of $Y_1$ and $Y_2$, respectively under $\hat{\Q}$, then the early exercise premium may be written as
{\small\begin{align}
\begin{split}
\label{eqn-PutCall-EarlyExercisePremium3}
& \tilde{V}^P(t,\tilde{S}(t),v(t))\\
& \qquad = \E_{\hat{\Q}}\left[\left.\int_t^T\left(q_1 e^{-q_1 s}\tilde{S}(s)-q_2 e^{-q_2 s}\right)\vm{1}(\calA(s))\dif s\right|\calF_t\right]\\
& \qquad \qquad -\tilde{\lambda}_1\E_{\hat{\Q}}\Bigg[\int_t^T\int_{-\infty}^{b(s,\tilde{S}(s),v(s))}\left(\tilde{V}^A\left(s,\tilde{S}(s)e^{y},v(s)\right)-\left(e^{-q_1 s}\tilde{S}(s)e^{y}-e^{-q_2 s}\right)\right)\\
& \hspace{100pt} \times G_1(y)\vm{1}(\calA(s))\dif y\dif s \bigg|\calF_t \Bigg]\\
& \qquad \qquad -\tilde{\lambda}_2\E_{\hat{\Q}}\Bigg[\int_t^T\int_{-b(s,\tilde{S}(s),v(s))}^\infty\left(\tilde{V}^A\left(s,\tilde{S}(s)e^{-y},v(s)\right)-\left(e^{-q_1 s}\tilde{S}(s)e^{-y}-e^{-q_2 s}\right)\right)\\
& \hspace{100pt} \times G_2(y)\vm{1}(\calA(s))\dif y\dif s \bigg|\calF_t \Bigg],
\end{split}
\end{align}}
where $b(s,\tilde{S}(s),v(s))\equiv \ln\left(B(s,v(s))e^{(q_1-q_2)s}/\tilde{S}(s)\right)$.
\end{rem}

Similar to the findings of \citet{Gukhal-2001}, \citet{ChiarellaZiogas-2004}, and \citet{CheangChiarellaZiogas-2013}, the early exercise premium \eqref{eqn-PutCall-EarlyExercisePremium} for our transformed problem can be further decomposed into a diffusion component (the positive term) and a jump component (the negative terms). However, unlike the early exercise premium derived by \citet{CheangChiarellaZiogas-2013} for an American call option under SVJD dynamics, our early exercise premium representation contains two jump terms. This is because $\tilde{S}(t)$ has two sources of jumps: the jumps in the price of the first asset (given by the counting measure $p(\dif y_1,\dif t)$) and the jumps in the num\'eraire process (given by the counting measure $p(\dif y_2,\dif t)$). Nonetheless, the interpretation remains the same: the diffusion term captures the discounted expected value of cash flows due to dividends when asset prices are in the stopping region and the jump terms capture the rebalancing costs incurred by the holder of the American exchange option when a jump instantaneously occurs in the price of either asset, causing $\tilde{S}(t)$ to jump back into the continuation region immediately after the option is exercised.\footnote{In this situation, the investor is unable to adjust the decision to exercise in response to the instantaneous jump in asset prices and is therefore vulnerable to the rebalancing cost described earlier. A similar phenomenon in the context of consumption-investment problems with transaction costs in a L\'evy-driven market is explored in greater technical detail by \citet{deValliere-2016}.} Figure \ref{fig-RebalancingCosts} illustrates the loss (captured by the difference in option value and the exercise value) incurred by the option holder when the asset yield ratio jumps back into the continuation region due to a jump, by a factor $e^{Y_1}$, in the price of the first asset. A similar graphical analysis holds if the price of the num\'eraire asset instantaneously jumps instead of the price of the first asset (in this case, the new asset yield ratio is $\tilde{S}(t)e^{-Y_2}$).

\begin{figure}
\centering
\includegraphics[width = 5in]{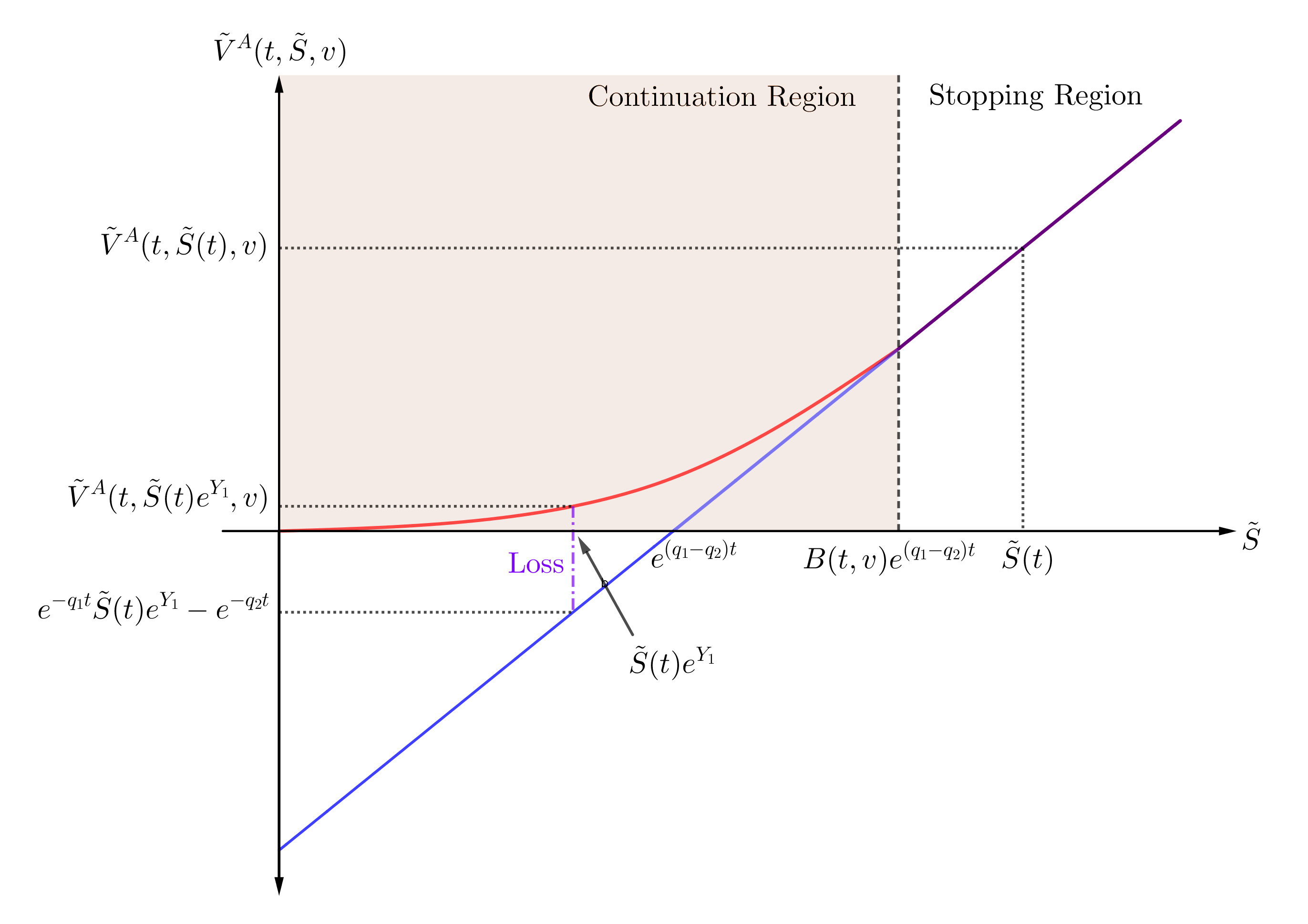}
\caption{Loss incurred by the option holder when the asset ratio instantaneously jumps from $\tilde{S}(t)$ in the stopping region to $\tilde{S}(t)e^{Y_1}$ back in the continuation region.}
\label{fig-RebalancingCosts}
\end{figure}


Recall that the discounted American exchange option $\tilde{V}^A(t,\tilde{S}(t),v(t))$ is a solution of the homogeneous IPDE $$\pder[\tilde{V}^A]{t}+\calL_{\tilde{s},v}\left[\tilde{V}^A(t,\tilde{S}(t),v(t))\right] = 0$$ over the restricted domain $0\leq t\leq T$, $0<\tilde{S}(t)<B(t,v)e^{(q_1-q_2)t}$, and $0<v<\infty$ subject to the value-matching condition \eqref{eqn-PutCall-ValueMatchingCondition-Vtilde}, the smooth-pasting condition \eqref{eqn-PutCall-SmoothPastingCondition-Vtilde}, and boundary conditions \eqref{eqn-PutCall-BoundaryConditions-Vtilde}. Following \citet{Jamshidian-1992} and \citet{ChiarellaZiogas-2004}, the restriction on the domain can be lifted by adding the appropriate inhomogeneous term to the IPDE such that the equation holds for $\tilde{S}(t)>0$.The inhomogeneous IPDE corresponding to our analysis is presented in the following proposition. This analysis requires that $\tilde{V}^A(t,\tilde{S}(t),v(t))$ and its first-order partial derivative with respect to $\tilde{S}$ are continuous, but the value-matching and smooth-pasting conditions are sufficient to meet this requirement.

\begin{prop}
\label{prop-PutCall-InhomogeneousTerm}
The discounted American exchange option price $\tilde{V}^A(t,\tilde{S},v)$ is a solution to the inhomogeneous IPDE
\begin{equation}
\label{eqn-PutCall-InhomogeneousIPDE}
0 = \pder[\tilde{V}^A]{t}+\calL_{\tilde{s},v}\left[\tilde{V}^A(t,\tilde{S}(t),v(t))\right]+\Xi(t,\tilde{S}(t),v(t)),
\end{equation}
where the inhomogeneous term $\Xi$ is given by
{\small\begin{align}
\begin{split}
\label{eqn-PutCall-InhomogeneousTerm}
& \Xi(t,\tilde{S}(t),v(t))\\
	& = \left(q_1 e^{-q_1 t}\tilde{S}(t)-q_2 e^{-q_2 t}\right)\vm{1}(\calA(t))\\
	& \quad - \tilde{\lambda}_1\vm{1}(\calA(t))\int_{-\infty}^{b(t,\tilde{S}(t),v(t))}\left[\tilde{V}^A\left(t,\tilde{S}(t)e^{y},v(t)\right)-\left(e^{-q_1 t}\tilde{S}(t)e^y-e^{-q_2 t}\right)\right]G_1(y)\dif y\\
	& \quad - \tilde{\lambda}_2\vm{1}(\calA(t))\int_{-b(t,\tilde{S}(t),v(t))}^{\infty}\left[\tilde{V}^A\left(t,\tilde{S}(t)e^{-y},v(t)\right)-\left(e^{-q_1 t}\tilde{S}(t)e^{-y}-e^{-q_2 t}\right)\right]G_2(y)\dif y,
\end{split}
\end{align}}
where $G_1$ and $G_2$ are the pdfs of $Y_1$ and $Y_2$, respectively, under $\hat{\Q}$, and $b(t,\tilde{S}(t),v(t)) \equiv \ln[B(t,v(t))e^{(q_1-q_2)t}/\tilde{S}(t)]$. This equation is to be solved for $(t,\tilde{S}(t),v(t))\in[0,T]\times\mathbb{R}_+^2$, subject to terminal and boundary conditions \eqref{eqn-PutCall-BoundaryConditions-Vtilde}.
\end{prop}

\begin{proof}
Observe that for all $(t,\tilde{S}(t),v(t))\in[0,T]\times\mathbb{R}_+^2$, the equation $$0 = \pder[\tilde{V}^A]{t}+\calL_{\tilde{s},v}\left[\tilde{V}^A(t,\tilde{S}(t),v(t))\right] - \left(\pder[\tilde{V}^A]{t}+\calL_{\tilde{s},v}\left[\tilde{V}^A(t,\tilde{S}(t),v(t))\right]\right)\vm{1}(\calA(t))$$ holds. Equation \eqref{eqn-PutCall-InhomogeneousTerm} is obtained by expanding the negative term in the above equation, as was done to obtain equation \eqref{eqn-PutCall-IPDE-StoppingRegion} and using the $G_1$ and $G_2$ to rewrite the expectations as integrals.
\end{proof}

\section{Limit of the Early Exercise Boundary at Maturity}
\sectionmark{Limit of the Early Exercise Boundary}
\label{sec-PutCall-EEBLimit}


Of particular interest is the behavior of the unknown early exercise boundary near the maturity of the option and the conditions on model parameters under which the boundary is continuous at maturity. The next proposition presents the limit of the early exercise boundary, which we obtain following the method of \citet{ChiarellaZiogas-2009}.

\begin{prop}
\label{prop-PutCall-EEBLimit}
The limit $B(T^-,v)\equiv\lim_{t\to T^-}B(t,v)$ is a solution of the equation
\begin{equation}
\label{eqn-PutCall-EEBLimit}
B(T^-,v) = \max\left\{1,\frac{q_2+\tilde{\lambda}_1\int_{-\infty}^{-\ln B(T^-,v)} G_1(y)\dif y + \tilde{\lambda}_2\int_{\ln B(T^-,v)}^\infty G_2(y)\dif y}{q_1 +\tilde{\lambda}_1\int_{-\infty}^{-\ln B(T^-,v)} e^y G_1(y)\dif y + \tilde{\lambda}_2\int_{\ln B(T^-,v)}^\infty e^{-y} G_2(y)\dif y}\right\}.
\end{equation}
\end{prop}

\begin{proof}
The method of \citet{ChiarellaZiogas-2009}, adapted to our situation, is as follows.\footnote{\citet{ChiarellaZiogas-2009} proposed this method as an alternative to the local analysis of the option PDE for small time-to-maturity options as was done by \citet{Wilmott-1993} in the pure diffusion case.} First, we set the inhomogeneous term $\Xi(t,\tilde{S},v)$ (given by equation \eqref{eqn-PutCall-InhomogeneousTerm}) to zero and evaluate the result at $t=T$ and $\tilde{S}=B(T^-,v)e^{(q_1-q_2)T}$. The resulting expression is then rearranged to yield equation \eqref{eqn-PutCall-EEBLimit}.

Performing the first step yields the equation
\begin{align}
\begin{split}
\label{eqn-PutCall-EEBLimit-Step1}
0 & = e^{-q_2 T}\left(q_1 B(T^-,v)-q_2\right)\\
	& \qquad - \tilde{\lambda}_1\int_{-\infty}^{-\ln\left[\frac{B(T,v)e^{(q_1-q_2)T}}{B(T^-,v)e^{(q_1-q_2)T}}\right]}\left[\tilde{V}^A\left(T, B(T^-,v)e^{(q_1-q_2)T}e^y, v(T)\right)\right.\\
	& \hspace{100pt} \left.-e^{-q_2 T}\left(B(T^-,v)e^{y}-1\right)\right]G_1(y)\dif y\\
	& \qquad - \tilde{\lambda}_2\int_{\ln\left[\frac{B(T,v)e^{(q_1-q_2)T}}{B(T^-,v)e^{(q_1-q_2)T}}\right]}^{\infty}\left[\tilde{V}^A\left(T, B(T^-,v)e^{(q_1-q_2)T}e^{-y}, v(T)\right)\right.\\
	& \hspace{100pt} \left.-e^{-q_2 T}\left(B(T^-,v)e^{-y}-1\right)\right]G_2(y)\dif y.
\end{split}
\end{align}
At maturity $t=T$, the option will be exercised if $\tilde{S}(T) \geq e^{(q_1-q_2)T}$, and so $B(T,v)=1$. Thus in the above calculation, setting $\tilde{S}=B(T^-,v)e^{(q_1-q_2)T}$ induces the stopping criterion in equation \eqref{eqn-PutCall-StoppingContinuationRegions2} for $\calS(T)$ since $B(T^-,v)\geq 1$. This implies that $\vm{1}(\calA(T))=1$ in the inhomogeneous term \eqref{eqn-PutCall-InhomogeneousTerm}. Furthermore, terminal condition \eqref{eqn-PutCall-BoundaryConditions-Vtilde} implies that
\begin{align*}
\tilde{V}^A\left(T,B(T^-,v)e^{(q_1-q_2)T}e^{y},v(T)\right) & = \max\left\{0,e^{-q_2 T}\left(B(T^-,v)e^y-1\right)\right\}\\
\tilde{V}^A\left(T,B(T^-,v)e^{(q_1-q_2)T}e^{-y},v(T)\right) & = \max\left\{0,e^{-q_2 T}\left(B(T^-,v)e^{-y}-1\right)\right\}.
\end{align*}
Thus, the first integral in equation \eqref{eqn-PutCall-EEBLimit-Step1} will be zero if $B(T^-,v)e^y-1\geq 0$ or if $y\geq -\ln B(T^-,v)$. Likewise, the second integral will be zero if $B(T^-,v)e^{-y}-1\geq 0$ or if $y\leq\ln B(T^-,v)$. In other words, the integral terms will vanish if the maximum functions yield the nonzero alternative. Following this analysis, equation \eqref{eqn-PutCall-EEBLimit-Step1} simplifies to
\begin{align*}
0 & = q_1 B(T^-,v)-q_2 +\tilde{\lambda}_1\int_{-\infty}^{-\ln B(T^-,v)}\left[B(T^-,v)e^{y}-1\right]G_1(y)\dif y\\
	& \qquad +\tilde{\lambda}_2\int_{\ln B(T^-,v)}^{\infty}\left[B(T^-,v)e^{-y}-1\right]G_2(y)\dif y.
\end{align*}
Rearranging the terms yields the equation $$B(T^-,v) = \frac{q_2+\tilde{\lambda}_1\int_{-\infty}^{-\ln B(T^-,v)} G_1(y)\dif y + \tilde{\lambda}_2\int_{\ln B(T^-,v)}^\infty G_2(y)\dif y}{q_1 +\tilde{\lambda}_1\int_{-\infty}^{-\ln B(T^-,v)} e^y G_1(y)\dif y + \tilde{\lambda}_2\int_{\ln B(T^-,v)}^\infty e^{-y} G_2(y)\dif y}.$$

We note lastly from \citet{BroadieDetemple-1997} that $B(t,v)\geq 1$ for any $t\in[0,T]$ and $v\in(0,\infty)$. Therefore, we must enforce a lower bound of 1 on $B(T^-,v)$ via the maximum function. The result stated in the proposition thus holds.
\end{proof}

Equation \eqref{eqn-PutCall-EEBLimit} must be solved implicitly for $B(T^-,v)$, which can be done using standard root-finding techniques. From our analysis, we find that the limit is dependent on the asset dividend yields $q_1$ and $q_2$, the jump intensities $\tilde{\lambda}_1$ and $\tilde{\lambda}_2$, and the jump size densities $G_1$ and $G_2$. These dependencies highlight the influence of jumps in asset prices on the limiting behavior of the early exercise boundary.\footnote{This is in contrast to the proposition of \citet{CarrHirsa-2003}, in their analysis of the one-asset American put option where the log-price is driven by a L\'evy process, that the limit of the early exercise boundary is only dependent on the dividend yield and the risk-free rate.} In the succeeding discussion, we investigate the more specific effects of these parameters on the limit of the early exercise boundary. We note further that equation \eqref{eqn-PutCall-EEBLimit} does not depend on the instantaneous variance $v$ since the option payoff is independent of $v$. However, equation \eqref{eqn-PutCall-EEBLimit} is true for all $v\in(0,\infty)$.

In the absence of jumps (i.e. when $\tilde{\lambda}_1=\tilde{\lambda}_2=0$), the limit reduces to $\max\{1,q_2/q_1\}$. This is consistent with the result of \citet{BroadieDetemple-1997} for American exchange options in the pure diffusion case. In the pure diffusion case, $B(T^-,v) = q_2/q_1 > 1$ if $q_2>q_1$, implying that the early exercise boundary many not be continuous in $t$ at maturity. When jumps are present, the analysis of continuity becomes more complicated, as shown below.

First, we present some conditions under which equation \eqref{eqn-PutCall-EEBLimit} has a solution.

\begin{prop}
\label{prop-PutCall-EEBLimit-Existence}
Suppose $q_1,q_2\geq 0$ and $\tilde{\lambda}_1,\tilde{\lambda}_2>0$ are given and let $G_1$ and $G_2$ be continious probability density functions. The equation
\begin{equation}
\label{eqn-PutCall-EEBLimit-ImplicitComponent}
x = \frac{q_2+\tilde{\lambda}_1\int_{-\infty}^{-\ln x} G_1(y)\dif y + \tilde{\lambda}_2\int_{\ln x}^\infty G_2(y)\dif y}{q_1 +\tilde{\lambda}_1\int_{-\infty}^{-\ln x} e^y G_1(y)\dif y + \tilde{\lambda}_2\int_{\ln x}^\infty e^{-y} G_2(y)\dif y}
\end{equation}
has a unique solution $x^*\in(0,\infty)$ if $q_1>0$. Furthermore, $x^*>1$ if and only if $$q_2-q_1+\tilde{\lambda}_1\int_{-\infty}^0 (1-e^y)G_1(y)\dif y+\tilde{\lambda}_2\int_0^\infty (1-e^{-y})G_2(y)\dif y > 0.$$
\end{prop}

\begin{proof}
Our proof adapts the arguments made by \citet[pp. 34-35]{ChiarellaKangMeyer-2015}. For $x\in(0,\infty)$, define the function
\begin{align*}
f(x)	& = q_2+\tilde{\lambda}_1\int_{-\infty}^{-\ln x} G_1(y)\dif y + \tilde{\lambda}_2\int_{\ln x}^\infty G_2(y)\dif y\\
			& \qquad - x\left(q_1 +\tilde{\lambda}_1\int_{-\infty}^{-\ln x} e^y G_1(y)\dif y + \tilde{\lambda}_2\int_{\ln x}^\infty e^{-y} G_2(y)\dif y\right).
\end{align*}
Denote by $x^*$ a zero of $f$ (i.e. $f(x^*)=0$) on $(0,\infty)$, if any exist.

Differentiating $f$ with respect to $x$ yields
\begin{align*}
f'(x) & = - \left(q_1 +\tilde{\lambda}_1\int_{-\infty}^{-\ln x} e^y G_1(y)\dif y + \tilde{\lambda}_2\int_{\ln x}^\infty e^{-y} G_2(y)\dif y\right).
\end{align*}
The integrals above are nonnegative. Hence, $f$ is strictly decreasing on $(0,\infty)$ if $q_1>0$. We also observe that $$\lim_{x\to 0^+}f(x) = q_2+\tilde{\lambda_1}\int_{-\infty}^\infty G_1(y)\dif y+\tilde{\lambda}_2\int_{-\infty}^\infty G_2(y)\dif y = q_2+\tilde{\lambda}_1+\tilde{\lambda}_2>0$$ and
\begin{align*}
\lim_{x\to\infty} f(x)
	& = q_1 + \tilde{\lambda}\lim_{x\to\infty}\int_{-\infty}^{-\ln x}G_1(y)\dif y+\tilde{\lambda}_2\lim_{x\to\infty}\int_{\ln x}^\infty G_2(y)\dif y\\
	& \qquad - \lim_{x\to\infty} x\left(q_1 +\tilde{\lambda}_1\int_{-\infty}^{-\ln x} e^y G_1(y)\dif y + \tilde{\lambda}_2\int_{\ln x}^\infty e^{-y} G_2(y)\dif y\right)\\
	& = q_2 -\lim_{x\to\infty} x q_1.
\end{align*}
Thus, if $q_1>0$, then $\lim_{x\to\infty}f(x)<0$ and so $f$ strictly decreases from positive to negative values as $x$ increases on $(0,\infty)$. Therefore, there exists a unique $x^*\in(0,\infty)$ such that $f(x^*)=0$.

Now suppose $q_1>0$. Evaluating $f$ at $x=1$ gives us $$f(1) = q_2-q_1+\tilde{\lambda}_1\int_{-\infty}^0 (1-e^y)G_1(y)\dif y+\tilde{\lambda}_2\int_0^\infty (1-e^{-y})G_2(y)\dif y.$$ If $f(1)\leq 0$, then $x^*$ must be in the interval $(0,1]$ since $f$ is strictly decreasing. Otherwise, $x^*>1$ if and only if $f(1)>0$, which is the condition stated in the proposition.
\end{proof}

If a solution $x^*$ of equation \eqref{eqn-PutCall-EEBLimit-ImplicitComponent} exists, then the limit of the early exercise boundary is $B(T^-,v)=\max\{1,x^*\}$.

\begin{rem}
If $q_1=0$ and $q_2>0$, then $f'(x)\leq 0$ for all $x\in(0,\infty)$ and $\lim_{x\to\infty}f(x) = q_2 > 0$. That is, $f$ is non-increasing and remains positive as $x$ increases in $(0,\infty)$. Thus, $f$ does not have any zeros on $(0,\infty)$ and hence it is not optimal to exercise the option prior to maturity.
\end{rem}

An immediate result from Proposition \ref{prop-PutCall-EEBLimit-Existence} is a condition for the continuity of $B(t,v)$ at maturity.

\begin{prop}
Suppose $q_1>0$. For any fixed $v\in(0,\infty)$, $B(t,v)$ is continuous at maturity $t=T$ if
\begin{equation}
\label{eqn-PutCall-EEBLimit-ContinuityCondition}
q_1\geq q_2+\tilde{\lambda}_1\int_{-\infty}^0 (1-e^y)G_1(y)\dif y+\tilde{\lambda}_2\int_0^\infty (1-e^{-y})G_2(y)\dif y.
\end{equation}
\end{prop}

\begin{proof}
Suppose $q_1>0$ and condition \eqref{eqn-PutCall-EEBLimit-ContinuityCondition} holds. Then from the discussion at the end of the proof of Proposition \ref{prop-PutCall-EEBLimit-Existence}, the solution $x^*$ to equation \eqref{eqn-PutCall-EEBLimit-ImplicitComponent} lies in the interval $(0,1]$. It follows that $B(T^-,v)=1$, which is also the value of $B(T,v)$. Thus, $B(t,v)$ is continuous at the option maturity.
\end{proof}

We briefly discuss the behavior of $B(T^-,v)$ with respect to changes in $q_1$. Note that $\partial f/\partial q_1 = -x<0$, so when $q_1$ decreases, $f(x)$ increases. In particular, for a given $q_1>0$, there exists $x^*\in(0,\infty)$ such that $f(x^*)=0$. If $q_1$ decreases, then $f(x^*)$ increases away from zero, thereby moving the unique zero of $f$ to some other number $x'\in(0,\infty)$ such that $x'>x^*$. In other words, the solution $x^*$ of equation \eqref{eqn-PutCall-EEBLimit-ImplicitComponent} increases without bound, and consequently $B(T^-,v)\to\infty$, as $q_1\to 0^+$. \emph{Thus, when the first asset bears no dividend yield, it is not optimal to exercise the American exchange option early or at least immediately prior to the option maturity.}

\section{The Transition Density Function}
\label{sec-PutCall-TransitionDensityFunction}

To determine the price of the American exchange option using equation \eqref{eqn-PutCall-EarlyExerciseRepresentation}, we need to solve for the price of the European exchange option $\tilde{V}$ and evaluate the early exercise premium $\tilde{V}^P$ in equation \eqref{eqn-PutCall-EarlyExercisePremium2}. To do so, we need to solve for the joint transition density function of the asset yield ratio process $\tilde{S}$ and the variance process $v$ under $\Q$. With the joint transition density function, we may evaluate the expectations in equations \eqref{eqn-PutCall-DiscountedExcOpPrice2} and \eqref{eqn-PutCall-EarlyExerciseRepresentation}.

Let $Q(T,u,b;t,s,v)$ denote the joint transition density function of $(\tilde{S},v)$ under the probability measure $\hat{\Q}$:
\begin{equation*}
Q(T,\tilde{s}_T,v_T;t,\tilde{s},v) = \hat{\Q}\left(\left.\tilde{S}(T)=\tilde{s}_T, v(T)=v_T\right|\tilde{S}(t) = \tilde{s}, v(t) = v\right).
\end{equation*}
This denotes the probability of passage from $(\tilde{S}(t),v(t))=(\tilde{s},v)$ at time $t$ to $(\tilde{s}_T,v_T)$ at time $T$. The Kolmogorov backward equation associated to $Q$ is given by
\begin{equation}
\label{eqn-PutCall-KolmogorovBE}
\pder[Q]{t}+\calL_{\tilde{s},v}\left[Q(T,\tilde{s}_T,v_T;t,\tilde{s},v)\right] = 0,
\end{equation}
which is to be solved for $t\in[0,T]$ and $(\tilde{s},v)\in\mathbb{R}_+^2$ subject to the terminal condition $Q(T,\tilde{s}_T,v_T;T,\tilde{s},v) = \delta(\tilde{s}-\tilde{s}_T)\delta(v-v_T),$ where $\delta(\cdot)$ is the Dirac-delta function. Since $T$, $\tilde{s}_T$, and $v_T$ are specific constants, we will denote $Q$ by $Q(t,\tilde{s},v)$ in the succeeding calculations for notational brevity.

To solve equation \eqref{eqn-PutCall-KolmogorovBE}, let $x = \ln\tilde{s}$ and define the function $H$ by
\begin{equation}
\label{eqn-PutCall-ChangeofVariable}
H(T,x_T,v_T;t,x,v) = Q(T,e^{x_T},v_t;t,e^x,v).
\end{equation}
Thus, when written in terms of $H$ and its partial derivatives, equation \eqref{eqn-PutCall-KolmogorovBE} becomes
\begin{align}
\begin{split}
\label{eqn-PutCall-KBE-ChangeofVariable}
0	& = \pder[H]{t}-\left(\tilde{\lambda}_1\tilde{\kappa}_1+\tilde{\lambda}_2\tilde{\kappa}_2^- + \frac{1}{2}\sigma^2 v\right)\pder[H]{x}+\left[\xi\eta-(\xi+\Lambda)v\right]\pder[H]{v}\\
	& \qquad + \frac{1}{2}\sigma^2 v\pder[^2H]{x^2}+\frac{1}{2}\omega^2 v\pder[^2H]{v^2}+\omega\left(\sigma_1\rho_1-\sigma_2\rho_2\right)v\pder[^2H]{x\partial v}\\
	& \qquad + \tilde{\lambda}_1\E_{\hat{\Q}}^{Y_1}\left[H(t,x+Y_1,v)-H(t,x,v)\right]\\
	& \qquad +\tilde{\lambda}_2\E_{\hat{\Q}}^{Y_2}\left[H(t,x-Y_2,v)-H(t,x,v)\right].
\end{split}
\end{align}
The associated terminal condition is $H(T,x_T,v_T;T,x,v) = \delta(x-x_T)\delta(v-v_T).$ Moving forward, we shall denote $H$ by $H(t,x,v)$ to emphasize that we are solving equation \eqref{eqn-PutCall-KBE-ChangeofVariable} in $t$, $x$, and $v$ and that $x_T$ and $v_T$ are given terminal values of these variables.

The function $H(t,x,v)$ may be interpreted as the joint transition density function of the process $(X(t),v(t))$ indicating the probability of passage from $(x,v)$ at time $t$ to $(x_T,v_T)$ at time $t=T$. In symbols, we have $$H(T,x_T,v_T;t,x,v) = \hat{\Q}\left(\left.X(T)=x_T, v(T)=v_T\right|X(t) = x, v(t) = v\right).$$

Since the coefficients of equation \eqref{eqn-PutCall-KBE-ChangeofVariable} no longer contain $x$, we can take its Fourier transform with respect to $x$ to further simplify the equation. The following technical assumption is required to be able to take the Fourier transform of the partial derivatives of $H$. This assumption is reasonable to impose on $H$ as it is a transition density function and is expected to vanish in the extremities of its domain \citep{ChiarellaZiogasZiveyi-2010, CheangChiarellaZiogas-2013}.

\begin{asp}
\label{asp-PutCall-FourierTransformAssumption}
As $x\to\pm\infty$, $H(t,x,v) \to 0$, $\partial H/\partial x \to 0$, and $\partial H/\partial v \to 0$.
\end{asp}

Given this assumption on $H$, we now take the Fourier transform of equation \eqref{eqn-PutCall-KBE-ChangeofVariable} in $x$.

\begin{prop}
\label{prop-PutCall-FourierTransformedPDE}
Let $\hat{H}(t,\phi,v)$ denote the Fourier transform of $H(t,x,v)$ with respect to $x$,
\begin{equation}
\label{eqn-PutCall-FourierTransform}
\hat{H}(t,\phi,v) = \scrF_x\left\{H(t,x,v)\right\}(\phi) = \int_{-\infty}^\infty e^{i\phi x}H(t,x,v)\dif x.
\end{equation}
Then $\hat{H}$ satisfies the equation
\begin{equation}
\label{eqn-PutCall-FourierTransformedPDE}
0 = \pder[\hat{H}]{t}+\left(\frac{1}{2}\varepsilon v-i\phi\Psi\right)\hat{H}+\left(\alpha-\Theta v\right)\pder[\hat{H}]{v}+\frac{1}{2}\omega^2 v\pder[^2\hat{H}]{v^2},
\end{equation}
where
\begin{align}
\begin{split}
\label{eqn-PutCall-FourierTransformedPDE-Parameters}
\alpha	& \equiv \xi\eta\\
\Theta = \Theta(\phi) & \equiv \xi+\Lambda+i\phi\omega\left(\sigma_1\rho_1-\sigma_2\rho_2\right)\\
\varepsilon = \varepsilon(\phi)	& \equiv \sigma^2\left(i\phi-\phi^2\right)\\
\Psi = \Psi(\phi) & \equiv -\tilde{\lambda}_1\tilde{\kappa}_1-\tilde{\lambda}_2\tilde{\kappa}_2^- - \frac{\tilde{\lambda}_1}{i\phi}\left(\varphi_1(\phi)-1\right)-\frac{\tilde{\lambda}_2}{i\phi}\left(\varphi_2(-\phi)-1\right),
\end{split}
\end{align}
and $\varphi_j(\phi) = \int_{\mathbb{R}}e^{-i\phi y}G_j(y)\dif y$ is the characteristic function of $Y_j$ under $\hat{\Q}$, $j=1,2$. The associated terminal condition is
\begin{equation}
\hat{H}(T,\phi,v) = e^{i\phi x_T}\delta(v-v_T).
\end{equation}
\end{prop}

\begin{proof}
See Appendix \ref{app-prop-PutCall-FourierTransformedPDE-Proof}.
\end{proof}

Save for some minor notational differences, PDE \eqref{eqn-PutCall-FourierTransformedPDE} is identical to the PDE presented in \citet[Proposition 4.1]{CheangChiarellaZiogas-2013} for the transition density function of a single-asset stochastic volatility jump-diffusion model. This resemblance is expected since the underlying asset yield ratio process $\tilde{S}(t)$ is modelled similarly with a Heston-type stochastic volatility process but with two jump components. The notational difference is pronounced in the definition of $\Psi(\phi)$ in equation \eqref{eqn-PutCall-FourierTransformedPDE-Parameters} where we have two terms corresponding to the jump size variables $Y_1$ and $Y_2$, in contrast to that of \citet{CheangChiarellaZiogas-2013} who only have one jump term. Due to these similarities, we follow the method of \citet{CheangChiarellaZiogas-2013} in the succeeding calculations to determine the solution of equation \eqref{eqn-PutCall-FourierTransformedPDE}.

At this point, we have reduced IPDE \eqref{eqn-PutCall-KBE-ChangeofVariable} to a second-order PDE \eqref{eqn-PutCall-FourierTransformedPDE} in $t$ and $v$. Solutions of second-order PDEs such as equation \eqref{eqn-PutCall-FourierTransformedPDE} have been obtained by \citet{Feller-1951} using a Laplace transform with respect to $v$. Thus we may further simplify equation \eqref{eqn-PutCall-FourierTransformedPDE} by taking its Laplace transform with respect to $v$. First, additional technical assumptions must be imposed on $\hat{H}(t,\phi,v)$ to ensure that all required Laplace transforms are well-defined.

\begin{asp}
\label{asp-PutCall-LaplaceTransformAssumption}
As $v\to+\infty$, $e^{-\vartheta v}\hat{H}(t,\phi,v) \to 0$ and $e^{-\vartheta v}\partial\hat{H}/\partial v \to 0$.
\end{asp}

As before, these technical conditions can be reasonably assumed since $H$ is a transition density function. In particular, the Fourier transform $\hat{H}$ defined in equation \eqref{eqn-PutCall-FourierTransform} and its partial derivative in $v$ both decay to zero since $H(t,x,v)$ decays to zero as $v\to\infty$. Assumption \ref{asp-PutCall-LaplaceTransformAssumption} also implies that the growth of $\hat{H}$ and $\partial\hat{H}/\partial v$ is dominated by the growth of the exponential term $e^{\vartheta v}$ as $v\to\infty$ for any $\vartheta>0$. Taking Assumption \ref{asp-PutCall-LaplaceTransformAssumption} to be true, we now reduce equation \eqref{eqn-PutCall-FourierTransformedPDE} to a first-order PDE in the following proposition.

\begin{prop}
\label{prop-PutCall-LaplaceTransformedPDE}
Let $\bar{H}(t,\phi,\vartheta)$ be the Laplace transform of $\hat{H}(t,\phi,v)$ with respect to $v$,
\begin{equation}
\label{eqn-PutCall-LaplaceTransform}
\bar{H}(t,\phi,\vartheta) = \scrL_v\left\{\hat{H}(t,\phi,v)\right\}(\vartheta) = \int_0^\infty e^{-\vartheta v}\hat{H}(t,\phi,v)\dif v.
\end{equation}
Then $\bar{H}$ satisfies the equation
\begin{equation}
\label{eqn-PutCall-LaplaceTransformedPDE}
-\pder[\bar{H}]{t} + \left[\frac{1}{2}\omega^2\vartheta^2-\Theta\vartheta+\frac{1}{2}\varepsilon\right]\pder[\bar{H}]{\vartheta} = \left[(\alpha-\omega^2)+\Theta-i\phi\Psi\right]\bar{H}+f(t),
\end{equation}
where $f(t) \equiv (\omega^2/2-\alpha)\hat{H}(t,\phi,0)$ must be determined such that 
\begin{equation}
\label{eqn-PutCall-LaplaceTransformedPDE-FinitenessCondition}
\lim_{\vartheta\to\infty}\bar{H}(t,\phi,\vartheta) = 0.
\end{equation} 
Equation \eqref{eqn-PutCall-LaplaceTransformedPDE} has terminal condition
\begin{equation}
\label{eqn-PutCall-LaplaceTransformedPDE-TerminalCondition}
\bar{H}(T,\phi,\vartheta) = \exp\left\{i\phi x_T-\vartheta v_T\right\}.
\end{equation}
\end{prop}

\begin{proof}
See Appendix \ref{app-prop-PutCall-LaplaceTransformedPDE-Proof}.
\end{proof}

In the next proposition, we present the solution of equation \eqref{eqn-PutCall-LaplaceTransformedPDE}.

\begin{prop}
\label{prop-PutCall-LaplaceTransformSolution}
The solution $\bar{H}(t,\phi,\vartheta)$ of equation \eqref{eqn-PutCall-LaplaceTransformedPDE} is given by
\begin{align}
\begin{split}
\label{eqn-PutCall-LaplaceTransformSolution}
\bar{H}(t,\phi,\vartheta)
	& = \exp\left\{\left[\frac{(\alpha-\omega^2)(\Theta-\digamma)}{\omega^2}+\Theta-i\phi\Psi\right](T-t)\right\}\\
	& \qquad \times \left[\frac{2\digamma}{(\omega^2\vartheta-\Theta+\digamma)(e^{\digamma(T-t)}-1)+2\digamma}\right]^{2-\frac{2\alpha}{\omega^2}}\\
	& \qquad \times \exp\left\{i\phi x_T-\left(\frac{\Theta-\digamma}{\omega^2}\right)v_T\right\}\\
	& \qquad \times \exp\left\{\frac{-2\digamma v_T(\omega^2\vartheta-\Theta+\digamma)e^{\digamma(T-t)}}{\omega^2\left[(\omega^2\vartheta-\Theta+\digamma)(e^{\digamma(T-t)}-1)+2\digamma\right]}\right\}\\
	& \qquad \times \Gamma\left(\frac{2\alpha}{\omega^2}-1;\beta(\phi,\vartheta;v_T)\right),
\end{split}
\end{align}
where
\begin{align}
\begin{split}
\label{eqn-PutCall-LaplaceTransformSolution-Parameters}
\digamma = \digamma(\phi) & \equiv \sqrt{\Theta^2(\phi)-\omega^2\varepsilon(\phi)}\\
\beta(\phi,\vartheta;v_T) & \equiv \frac{2\digamma v_T e^{\digamma(T-t)}}{\omega^2 (e^{\digamma(T-t)}-1)}\times\frac{2\digamma}{(\omega^2\vartheta-\Theta+\digamma)(e^{\digamma(T-t)}-1)+2\digamma},
\end{split}
\end{align}
$\Gamma(u;\beta)$ is the (lower) incomplete gamma function $\Gamma(u;\beta) = \frac{1}{\Gamma(u)}\int_0^\beta e^{-x}x^{u-1}\dif x,$ and $\Gamma(u)$ is the gamma function $\Gamma(u) = \int_0^\infty e^{-x}x^{u-1}\dif x$.
\end{prop}

\begin{proof}
PDE \eqref{eqn-PutCall-LaplaceTransformedPDE} is an inhomogeneous first-order equation which can be solved via the method of characteristics and the method of variation of parameters. This procedure yields a solution $\bar{H}$ in terms of the unknown function $f(t)$. Condition \eqref{eqn-PutCall-LaplaceTransformedPDE-FinitenessCondition} is then applied to determine $f(t)$ and characterize $\bar{H}$ completely in terms of the parameters introduced in Propositions \ref{prop-PutCall-FourierTransformedPDE} and \ref{prop-PutCall-LaplaceTransformedPDE}.

Due to close resemblances in the form of the PDE, the derivation of equation \eqref{eqn-PutCall-LaplaceTransformSolution} follows the proof presented in \citet[Appendix 4]{ChiarellaZiogasZiveyi-2010}.\footnote{Note however that \citet{ChiarellaZiogasZiveyi-2010} uses time-to-maturity $\tau\equiv T-t$ instead of calendar time $t$, as what was done in our analysis and by \citet{CheangChiarellaZiogas-2013}.} For easy comparison between our analysis and the proofs presented in \citet{ChiarellaZiogasZiveyi-2010} and \citet{CheangChiarellaZiogas-2013}, we show in Table \ref{tab-ParameterComparison} the equivalence of notations used in the PDE. 
\end{proof}

\begin{table}
\begin{center}
\begin{tabular}{@{}ccc@{}}
\toprule
\textbf{Equation \eqref{eqn-PutCall-LaplaceTransformedPDE}} & \textbf{[CZZ-2010]} & \textbf{[CCZ-2013]} \\ \midrule
$\omega^2$ & $\sigma^2$ & $\sigma^2$ \\
$\varepsilon$ & $\Lambda$ & $\Lambda$ \\
$\vartheta$ & $s$ & $s$ \\
$-i\phi\Psi$ & $-i\phi(r-q)$ & $-i\phi\Psi$ \\
$T-t$ & $\tau$ & $T-t$ \\ 
\bottomrule
\end{tabular}
\end{center}
\caption{Comparison of notation and coefficients of equation \eqref{eqn-PutCall-LaplaceTransformedPDE}, \citet[equation 29]{ChiarellaZiogasZiveyi-2010} [CZZ-2010] and \citet[equation 46]{CheangChiarellaZiogas-2013} [CCZ-2013].}
\label{tab-ParameterComparison}
\end{table}

At this point, we now recover the original transition density function $Q(t,\tilde{s},v)$ by inverting the Laplace and Fourier transforms on $\bar{H}(t,\phi,\vartheta)$ given in equation \eqref{eqn-PutCall-LaplaceTransformSolution}. To this end, we first solve for the inverse Laplace transform of $\bar{H}$.

\begin{prop}
\label{prop-PutCall-InverseLaplaceTransform}
The inverse Laplace transform of $\bar{H}(t,\phi,\vartheta)$ in equation \eqref{eqn-PutCall-LaplaceTransformSolution} is
\begin{align}
\begin{split}
\label{eqn-PutCall-InverseLaplaceTransform}
\hat{H}(t,\phi,v)
	& = \exp\left\{\frac{(\Theta-\digamma)}{\omega^2}(v-v_T+\alpha(T-t))\right\} \exp\left\{i\phi x_T - i\phi\Psi(T-t)\right\}\\
	& \qquad \times \frac{2\digamma e^{\digamma(T-t)}}{\omega^2(e^{\digamma(T-t)}-1)}\left[\frac{v_T e^{\digamma(T-t)}}{v}\right]^{\frac{\alpha}{\omega^2}-\frac{1}{2}} \exp\left\{-\frac{2\digamma (v_T e^{\digamma(T-t)}+v)}{\omega^2(e^{\digamma(T-t)}-1)}\right\}\\
	& \qquad \times I_{\frac{2\alpha}{\omega^2}-1}\left(\frac{4\digamma\sqrt{v_T v e^{\digamma(T-t)}}}{\omega^2(e^{\digamma(T-t)}-1)}\right),
\end{split}
\end{align}
where $I_k(u)$ is the modified Bessel function of the first kind
\begin{equation}
I_k(u) = \sum_{n=0}^\infty \frac{(u/2)^{2n+k}}{n!\Gamma(n+k+1)}.
\end{equation}
\end{prop}

\begin{proof}
Refer to \citet[Appendix 5]{ChiarellaZiogasZiveyi-2010}, keeping in mind the notational equivalence established in Table \ref{tab-ParameterComparison}.
\end{proof}

Having determined $\hat{H}(t,\phi,v)$, we now invert the Fourier transform to recover $H(t,x,v)$ using the inversion formula
\begin{equation}
\label{eqn-PutCall-InverseFourierTransform}
H(t,x,v) = \scrF_x^{-1}\left\{\hat{H}(t,\phi,v)\right\}(x) = \frac{1}{2\pi}\int_{-\infty}^\infty e^{-i\phi x}\hat{H}(t,\phi,v)\dif\phi
\end{equation}
that corresponds to the Fourier transform defined by equation \eqref{eqn-PutCall-FourierTransform}. The original transition density function $Q(t,\tilde{s},v)$ is then obtained by reversing the substitution $x=\ln\tilde{s}$ made in equation \eqref{eqn-PutCall-ChangeofVariable}. The result is presented in the proposition below.

\begin{prop}
\label{prop-PutCall-TransitionDensityFunction2}
The transition density function $Q(t,\tilde{s},v) \equiv Q(T,\tilde{s}_T,v_T;t,\tilde{s},v)$ is given by
\begin{align}
\begin{split}
\label{eqn-PutCall-TransitionDensityFunction2}
Q(t,\tilde{s},v)
	& = \sum_{m=0}^\infty\sum_{n=0}^\infty\frac{(\tilde{\lambda}_1(T-t))^m(\tilde{\lambda}_2(T-t))^n e^{-(\tilde{\lambda}_1+\tilde{\lambda}_2)(T-t)}}{m!n!}\\
	& \qquad \times \E_{\hat{\Q}}^{(m,n)}\Bigg[\frac{1}{2\pi}\int_{-\infty}^\infty \exp\left\{-i\phi\ln\left(\frac{\tilde{s}}{\tilde{s}_T}e^{-(\tilde{\lambda}_1\tilde{\kappa}_1+\tilde{\lambda}_2\tilde{\kappa}_2^-)(T-t)}e^{\Upsilon_{1,m}-\Upsilon_{2,n}}\right)\right\}\\
	& \hspace{100pt} \times h(T-t,\phi,v;v_T)\dif\phi\Bigg],
\end{split}
\end{align}
where
\begin{align}
\begin{split}
\label{eqn-PutCall-TransitionDensityFunction2-h}
h(\tau,\phi,v;v_T)
	& = \exp\left\{\frac{(\Theta-\digamma)}{\omega^2}(v-v_T+\alpha\tau)\right\} \frac{2\digamma e^{\digamma\tau}}{\omega^2(e^{\digamma\tau}-1)}\left[\frac{v_T e^{\digamma\tau}}{v}\right]^{\frac{\alpha}{\omega^2}-\frac{1}{2}}\\
	& \qquad \times \exp\left\{-\frac{2\digamma (v_T e^{\digamma\tau}+v)}{\omega^2(e^{\digamma\tau}-1)}\right\} \times I_{\frac{2\alpha}{\omega^2}-1}\left(\frac{4\digamma\sqrt{v_T v e^{\digamma\tau}}}{\omega^2(e^{\digamma\tau}-1)}\right).
\end{split}
\end{align}
Here, $\Upsilon_{1,m}$ and $\Upsilon_{1,n}$ are given by $$\Upsilon_{1,m} = \sum_{k=1}^m Y_{1,k} \qquad \text{and} \qquad \Upsilon_{2,n} = \sum_{l=1}^n Y_{2,l},$$ where $\{Y_{1,1},\dots,Y_{1,m}\}$ and $\{Y_{2,1},\dots,Y_{2,n}\}$ are collections of i.i.d. random variables sampled from populations with $\hat{Q}$-density functions $G_1(y)$ and $G_2(y)$, respectively, of $Y_1$ and $Y_2$, and $\E_{\hat{\Q}}^{(m,n)}[\cdot]$ is the expectation operator with respect to $\Upsilon_{1,m}$ and $\Upsilon_{2,n}$ only.
\end{prop}

\begin{proof}
See Appendix \ref{app-prop-PutCall-TransitionDensityFunction2-Proof}.
\end{proof}

\section{The European Exchange Option Price}
\label{sec-PutCall-EuExcOpPrice}

Given the transition density function $Q(T,\tilde{s}_T,v_T;t,\tilde{s},v)$, we can now compute the price of any European-style option on the two assets which matures at time $T$ and whose payoff, when discounted by the second asset yield process, can be written as a function $\tilde{F}(T,\tilde{s}_T,v_T)$ of the terminal asset yield ratio $\tilde{S}(T) = \tilde{s}_T$ and the terminal instantaneous variance $v(T) = v_T$. Following equation \eqref{eqn-PutCall-DiscountedExcOpPrice} and the results stated in Proposition \ref{prop-PutCall-DiscountedEuExcOpPrice}, the no-arbitrage price at time $t\in[0,T)$ of such a claim is given by
\begin{align*}
P(t,\tilde{s},v) 
	& = S_2(t) e^{q_2 t} \int_0^\infty \int_0^\infty \tilde{F}(T,\tilde{s}_T,v_T)Q(T,\tilde{s}_T,v_T;t,\tilde{s},v)\dif v_T \dif x_T.
\end{align*}
In this section, we calculate the price of the European exchange option following the valuation formula stated above.

For our calculations, it is more convenient to use the log-price variable $x_T = \ln\tilde{s}_T$ and the corresponding transition density function $H(T,x_T,v_T;t,\tilde{s}_T,v_T)$. This implies that the option price $\tilde{V}(t,\tilde{s},v)$, as given by equation \eqref{eqn-PutCall-DiscountedExcOpPrice2}, can be written as\footnote{From this point forward, we make the simplifying notation $\tilde{S}(t) = \tilde{s}$ and $v(t) = v$ to be consistent with our notation for the transition density function.}
\begin{equation}
\label{eqn-PutCall-EuExcOpPrice}
\tilde{V}(t,\tilde{s},v) = e^{-q_1 T}\int_{-\infty}^\infty \int_0^\infty \left(e^{x_T}-e^{(q_1-q_2)T}\right)^+ H(T,x_T,v_T;t,\ln\tilde{s},v)\dif v_T \dif x_T.
\end{equation}
Since the payoff is independent of the terminal variance level $v_T$, we may evaluate the integral with respect to $v_T$ by evaluating $\int_0^\infty H(T,x_T,v_T;t,x,v)\dif v_T$. We present this integral in the following lemma.

\begin{lem}
\label{lem-PutCall-HTDFIntegral}
The integral $\int_0^\infty H(T,x_T,v_T;t,\ln\tilde{s},v)\dif v_T$ is given by
\begin{align}
\begin{split}
\label{eqn-PutCall-HTDFIntegral}
& \int_0^\infty H(T,x_T,v_T;t,\ln\tilde{s},v)\dif v_T\\
& \qquad = \sum_{m=0}^\infty \sum_{n=0}^\infty \frac{(\tilde{\lambda}_1(T-t))^m (\tilde{\lambda}_2(T-t))^n e^{-(\tilde{\lambda}_1+\tilde{\lambda}_2)(T-t)}}{m!n!}\\
& \qquad \qquad \times \E_{\hat{\Q}}^{(m,n)}\left[\frac{1}{2\pi}\int_{-\infty}^\infty e^{i\phi x_T}f\left(T-t,\tilde{s}e^{-(\tilde{\lambda}_1\tilde{\kappa}_1+\tilde{\lambda}_2\tilde{\kappa}_2^-)(T-t)}e^{\Upsilon_{1,m}-\Upsilon_{2,n}}, v;-\phi\right)\dif\phi\right],
\end{split}
\end{align}
where
\begin{align}
\begin{split}
\label{eqn-PutCall-HTDFIntegral-Parameters}
f(\tau,z,v;\phi) & = \exp\left\{i\phi\ln z+B(\tau,-\phi)+D(\tau,-\phi)v\right\}\\
B(\tau,\phi) & = \frac{\alpha}{\omega^2}\left\{(\Theta+\digamma)\tau-2\ln\left(\frac{1-\chi e^{\digamma\tau}}{1-\chi}\right)\right\}\\
D(\tau,\phi) & = \frac{\Theta+\digamma}{\omega^2}\left(\frac{1-e^{\digamma\tau}}{1-\chi e^{\digamma\tau}}\right),
\end{split}
\end{align}
with $\chi = (\Theta+\digamma)/(\Theta-\digamma)$, $\alpha$ and $\Theta$ as defined in Proposition \ref{prop-PutCall-FourierTransformedPDE}, and $\digamma$ as defined in Proposition \ref{prop-PutCall-LaplaceTransformSolution}.  
\end{lem}

\begin{proof}
Most of the proof deals with the evaluation of $\int_0^\infty h(T-t,\phi,v;v_T)\dif v_T$, where $h$ is given by equation \eqref{eqn-PutCall-TransitionDensityFunction2-h}. The steps in this integration closely follow those discussed in \citet[Appendix 6]{ChiarellaZiogasZiveyi-2010}.
\end{proof}

We are now ready to solve for the European exchange option price. Our results are shown in the next proposition.

\begin{prop}
\label{prop-PutCall-DiscEuExcOpPrice}
The discounted price $\tilde{V}(t,\tilde{s},v)$ of the European exchange option at time $t$ is given by
\begin{align}
\begin{split}
\label{eqn-PutCall-DiscEuExcOpPrice}
& \tilde{V}(t,\tilde{s},v) \\
	& \qquad = \sum_{m=0}^\infty \sum_{n=0}^\infty \frac{(\tilde{\lambda}_1(T-t))^m (\tilde{\lambda}_2(T-t))^n e^{-(\tilde{\lambda}_1+\tilde{\lambda}_2)(T-t)}}{m! n!}\\
	& \qquad \qquad \times \E_{\hat{\Q}}^{(m,n)}\Bigg[e^{-q_1 T} \tilde{s}e^{-(\tilde{\lambda}_1\tilde{\kappa}_1+\tilde{\lambda}_2\tilde{\kappa}_2^-)(T-t)}e^{\Upsilon_{1,m}-\Upsilon_{2,n}}\\
	& \hspace{75pt} \times P_1^E\left(T-t,\tilde{s}e^{-(\tilde{\lambda}_1\tilde{\kappa}_1+\tilde{\lambda}_2\tilde{\kappa}_2^-)(T-t)}e^{\Upsilon_{1,m}-\Upsilon_{2,n}},v;(q_1-q_2)T\right)\\
	& \hspace{75pt} - e^{-q_2 T} P_2^E\left(T-t,\tilde{s}e^{-(\tilde{\lambda}_1\tilde{\kappa}_1+\tilde{\lambda}_2\tilde{\kappa}_2^-)(T-t)}e^{\Upsilon_{1,m}-\Upsilon_{2,n}},v;(q_1-q_2)T\right)\Bigg]
\end{split}
\end{align}
where $P_1^E$ and $P_2^E$ are defined as
\begin{align}
\begin{split}
\label{eqn-PutCall-DiscEuExcOpPrice-Parameters}
P_1^E(\tau,z,v;K)	& = \frac{1}{2}+\frac{1}{2\pi}\int_0^\infty\frac{f_1(\tau,z,v;\phi)e^{-i\phi K}-f_1(\tau,z,v;-\phi)e^{i\phi K}}{i\phi}\dif\phi\\
P_2^E(\tau,z,v;K)	& = \frac{1}{2}+\frac{1}{2\pi}\int_0^\infty\frac{f(\tau,z,v;\phi)e^{-i\phi K}-f(\tau,z,v;-\phi)e^{i\phi K}}{i\phi}\dif\phi,
\end{split}
\end{align}
with $f(\tau,z,v;\phi)$ given in equation \eqref{eqn-PutCall-HTDFIntegral-Parameters}, $f_1(\tau,z,v;\phi)$ given by
\begin{align}
\begin{split}
\label{eqn-PutCall-DiscEuExcOpPrice-Parameters2}
f_1(\tau,z,v;\phi) & = \exp\left\{i\phi\ln z + B_1(\tau,-\phi)+D_1(\tau,-\phi)\right\}\\
B_1(\tau,\phi) & = \frac{\alpha}{\omega^2}\Bigg\{(\Theta_1+\digamma_1)\tau-2\ln\left[\frac{1-\chi_1 e^{\digamma_1\tau}}{1-\chi_1}\right]\Bigg\}\\
D_1(\tau,\phi) & = \frac{\Theta_1+\digamma_1}{\omega^2}\left[\frac{1-e^{\digamma_1\tau}}{1-\chi_1 e^{\digamma_1\tau}}\right],
\end{split}
\end{align}
$\Theta_1(\phi) \equiv \Theta(\phi-i)$, $\digamma_1(\phi) \equiv \digamma(\phi-i)$, and $\chi_1(\phi) = \chi(\phi-i)$.
\end{prop}

\begin{proof}
See Appendix \ref{app-prop-PutCall-DiscEuExcOpPrice-Proof}.
\end{proof}

The integrals that appear in equations \eqref{eqn-PutCall-DiscEuExcOpPrice} and \eqref{eqn-PutCall-DiscEuExcOpPrice-Parameters} can be evaluated using standard numerical integration schemes. Details of this implementation will no longer be discussed in this paper.

Equation \eqref{eqn-PutCall-DiscEuExcOpPrice} expresses the discounted European exchange option price as a sum of Poisson-weighted expectations of Heston-Bates-type formulas which arise when asset price dynamics have stochastic volatility. Our results are similar to those of \citet{CheangChiarellaZiogas-2013}, except that in the present analysis, there are two jump components considered. Equation \eqref{eqn-PutCall-DiscEuExcOpPrice} may also be read as an expected value of these Heston-Bates-type expressions conditional on the number of jumps in the prices of the two assets (denoted by $m$ and $n$) observed over the remaining life of the option. Given the number of jumps $m$ and $n$, the Heston-Bates expressions are then averaged with respect to the accumulated jump sizes $\Upsilon_{1,m}$ and $\Upsilon_{2,n}$ in the prices of stocks 1 and 2, respectively. Our representation for the discounted European exchange option price depends on the current asset yield ratio $\tilde{s}$, the dividend yields of the two assets, the jump parameters under $\hat{\Q}$, the accumulation of jumps $\Upsilon_{1,m}$ and $\Upsilon_{2,n}$, and the instantaneous variance level $v$. Notably absent is the risk-free rate $r$, which is consistent with the original result of \citet{Margrabe-1978}. The expectation $\E_{\hat{\Q}}^{(m,n)}$ may be evaluated further by specifying the jump-size distribution (e.g. a log-normal distribution as in \citet{Merton-1976} or an asymmetric double exponential distribution \`a la \citet{Kou-2002}), but we leave this unspecified for now.
 
The time $t$ price of the European exchange option price $C(t,S_1,S_2,v)$ may be obtained by multiplying $\tilde{V}(t,\tilde{s},v)$ by $S_2 e^{q_2 t}$. Doing so, we can write
\begin{equation}
\label{eqn-PutCall-EuExcOpPrice3}
C(t,S_1,S_2,v) = S_1 e^{-q_1(T-t)}\hat{Q}_1 - S_2 e^{-q_2 (T-t)}\hat{Q}_2,
\end{equation}
where
\begin{align}
\begin{split}
\label{eqn-PutCall-EuExcOpPrice3-Parameters}
\hat{Q}_1	
	& = \sum_{m=0}^\infty \sum_{n=0}^\infty \frac{(\tilde{\lambda}_1(T-t))^m (\tilde{\lambda}_2(T-t))^n e^{-(\tilde{\lambda}_1+\tilde{\lambda}_2)(T-t)}}{m! n!}\\
	& \qquad \times \E_{\hat{\Q}}^{(m,n)}\Bigg[e^{-(\tilde{\lambda}_1\tilde{\kappa}_1+\tilde{\lambda}_2\tilde{\kappa}_2^-)(T-t)}e^{\Upsilon_{1,m}-\Upsilon_{2,n}}\\
	& \hspace{75pt} \times P_1^E\left(T-t,\tilde{s}e^{-(\tilde{\lambda}_1\tilde{\kappa}_1+\tilde{\lambda}_2\tilde{\kappa}_2^-)(T-t)}e^{\Upsilon_{1,m}-\Upsilon_{2,n}},v;(q_1-q_2)T\right) \Bigg]\\
\hat{Q}_2
	& = \sum_{m=0}^\infty \sum_{n=0}^\infty \frac{(\tilde{\lambda}_1(T-t))^m (\tilde{\lambda}_2(T-t))^n e^{-(\tilde{\lambda}_1+\tilde{\lambda}_2)(T-t)}}{m! n!}\\
	& \qquad \times \E_{\hat{\Q}}^{(m,n)}\Bigg[P_2^E\left(T-t,\tilde{s}e^{-(\tilde{\lambda}_1\tilde{\kappa}_1+\tilde{\lambda}_2\tilde{\kappa}_2^-)(T-t)}e^{\Upsilon_{1,m}-\Upsilon_{2,n}},v;(q_1-q_2)T\right)\Bigg].
\end{split}
\end{align}
This form emphasizes the similarity of our result to the original \citet{Margrabe-1978} result and to the \citet{CheangChiarella-2011} extension to the jump-diffusion case. Furthermore, this representation is also consistent with the form obtained in \citet[equation 27]{CheangGarces-2019} which was obtained using the change-of-num\'eraire techniqueof \citet{Geman-1995}. In their analysis, $\hat{Q}_1$ and $\hat{Q}_2$ were interpreted as probabilities that the exchange option is in-the-money at maturity, taken under probability measures $\Q_1$ and $\Q_2$ equivalent to the risk-neutral probability measure $\Q$ (corresponding to the choice of the money-market account as the num\'eraire). The alternative measures $\Q_1$ and $\Q_2$ were obtained by taking the processes $\{S_1(t)e^{-(r-q_1)t}/S_1(0)\}$ and $\{S_2(t)e^{-(r-q_2)t}/S_2(0)\}$, respectively, as the num\'eraire process.\footnote{We refer to the process $\{S_i(t)e^{-(r-q_i)t}\}$ as the discounted yield process of stock $i$, $i=1,2$. As such, the probability measure $\hat{\Q}$ used in this analysis is different from $\Q_2$ as the former corresponds to the asset yield process $S_2(t)e^{q_2 t}$ as the num\'eraire. However, $\hat{\Q}$, $\Q_1$, and $\Q_2$ are equivalent probability measures as they are all equivalent to the objective market measure $\prob$. In this paper, we chose to perform our calculations under $\hat{\Q}$ as it reduces the exchange option pricing problem to the pricing of an ordinary call option on the asset yield process $\tilde{S}(t)$.}

\section{The American Exchange Option Price}
\sectionmark{Early Exercise Premium}
\label{sec-PutCall-EEP}

Knowledge of the transition density function also allows us to evaluate the expectations occurring in the early exercise premium $\tilde{V}^P(t,\tilde{s},v)$ given in equation \eqref{eqn-PutCall-EarlyExercisePremium2}. The starting point of our calculations is equation \eqref{eqn-PutCall-EarlyExercisePremium3} which expresses the early exercise premium in terms of the jump size density functions. In the succeeding calculations, let $\tilde{S}(u) = \tilde{s}_u$, $v(u) = v_u$, and $x_u = \ln \tilde{s}_u$, $u\in[t,T]$. The following proposition provides an integral representation for the early exercise premium.

\begin{prop}
\label{prop-PutCall-EarlyExercisePremium-IntegralRepresentation}
The discounted American exchange option price is given by $$\tilde{V}^A(t,\tilde{s},v) = \tilde{V}(t,\tilde{s}v)+\tilde{V}^P(t,\tilde{s},v),$$ where $\tilde{V}$ is the European exchange option price given in Proposition \ref{prop-PutCall-DiscEuExcOpPrice} and $\tilde{V}^P$ is the early exercise premium given by
\begin{equation}
\label{eqn-PutCall-earlyExercisePremium4}
\tilde{V}^P(t,\tilde{s},v) = \tilde{V}^P_D(t,\tilde{s},v)-\tilde{\lambda}_1\tilde{V}^P_{J_1}(t,\tilde{s},v)-\tilde{\lambda}_2\tilde{V}^P_{J_2}(t,\tilde{s},v).
\end{equation}
Here, $\tilde{V}^P_D$ is given by
\begin{align}
\begin{split}
\label{eqn-PutCall-EarlyExercisePremium-Diffusion}
& \tilde{V}^P_D(t,\tilde{s},v)\\
	& = \sum_{m=0}^\infty \sum_{n=0}^\infty \frac{\tilde{\lambda}_1^m \tilde{\lambda}_2^n}{m!n!}\E_{\hat{\Q}}^{(m,n)}\Bigg[\int_0^T \int_0^\infty (u-t)^{m+n}e^{-(\tilde{\lambda}_1+\tilde{\lambda}_2)(u-t)}\\
	& \hspace{25pt} \times \Bigg(q_1 e^{-q_1 u}\tilde{s}e^{-(\tilde{\lambda}_1\tilde{\kappa}_1+\tilde{\lambda}_2\tilde{\kappa}_2^-)(u-t)}e^{\Upsilon_{1,m}-\Upsilon_{2,n}}\\
	& \hspace{50pt} \times P_1^A\left(u-t,\tilde{s}e^{-(\tilde{\lambda}_1\tilde{\kappa}_1+\tilde{\lambda}_2\tilde{\kappa}_2^-)(u-t)}e^{\Upsilon_{1,m}-\Upsilon_{2,n}},v;v_u,A(u,v_u)\right)\\
	& \hspace{25pt} - q_2 e^{-q_2 u}\\
	& \hspace{50pt} \times P_2^A\left(u-t,\tilde{s}e^{-(\tilde{\lambda}_1\tilde{\kappa}_1+\tilde{\lambda}_2\tilde{\kappa}_2^-)(u-t)}e^{\Upsilon_{1,m}-\Upsilon_{2,n}},v;v_u,A(u,v_u)\right)\Bigg)\dif v_u \dif u\Bigg],
\end{split}
\end{align}
where $A(u,v_u) \equiv B(u,v_u)e^{(q_1-q_2)u}$ is the critical asset yield ratio associated to the stopping region $\calS(u)$ (see equation \eqref{eqn-PutCall-StoppingContinuationRegions2}), and $P_1^A$ and $P_2^A$ are defined as
{\small\begin{align}
\begin{split}
\label{eqn-PutCall-EarlyExercisePremium-Diffusion-Parameters}
P_1^A(\tau,z,v;v_u,K) & = \frac{1}{2}+\frac{1}{2\pi}\int_0^\infty \frac{f_{21}(\tau,z,v;\phi,v_u)e^{-i\phi\ln K}-f_{21}(\tau,z,v;-\phi,v_u)e^{i\phi\ln K}}{i\phi}\dif\phi\\
P_2^A(\tau,z,v;v_u,K) & = \frac{1}{2}+\frac{1}{2\pi}\int_0^\infty \frac{f_2(\tau,z,v;\phi,v_u)e^{-i\phi\ln K}-f_2(\tau,z,v;-\phi,v_u)e^{i\phi\ln K}}{i\phi}\dif\phi,
\end{split}
\end{align}}
with $f_2(\tau,z,v;\phi,v_u) = e^{i\phi\ln z}h(\tau,-\phi,v;v_u)$, $f_{21}$ given by $$f_{21}(\tau,z,v;\phi,v_u) = e^{i\phi\ln z}h(\tau,-(\phi-i),v;v_u),$$ and $h(\tau,\phi,v;v_u)$ given by equation \eqref{eqn-PutCall-TransitionDensityFunction2-h}. Furthermore, $\tilde{V}^P_{J_1}(t,\tilde{s},v)$ and $\tilde{V}^P_{J_1}(t,\tilde{s},v)$ are given by
\begin{align}
\begin{split}
\label{eqn-PutCall-EarlyExercisePremium-J1}
& \tilde{V}^P_{J_1}(t,\tilde{s},v)\\
& \qquad = \sum_{m=0}^\infty \sum_{n=0}^\infty \frac{\tilde{\lambda}_1^m \tilde{\lambda}_2^n}{m!n!}\E_{\hat{\Q}}^{(m,n)}\Bigg[\frac{1}{2\pi}\int_0^T \int_0^\infty (u-t)^{m+n}e^{-(\tilde{\lambda}_1+\tilde{\lambda}_2)(u-t)}\\
& \hspace{25pt} \times \int_{-\infty}^0 G_1(y) \int_{\ln A(u,v_u)}^{\ln A(u,v_u)-y} \left(V^A(u,e^{x_u+y},v_u)-\left(e^{-q_1 u}e^{x_u+y}-e^{-q_2 u}\right)\right)\\
& \hspace{25pt} \times \int_{-\infty}^\infty e^{i\phi x_u}\exp\left\{-i\phi\ln\left(\tilde{s}e^{-(\tilde{\lambda}_1\tilde{\kappa}_1+\tilde{\lambda}_2\tilde{\kappa}_2^-)(u-t)}e^{\Upsilon_{1,m}-\Upsilon_{2,n}}\right)\right\}\\
& \hspace{25pt} \times h(u-t,\phi,v;v_u)\dif\phi \dif x_u \dif y \dif v_u \dif u \Bigg],
\end{split}
\end{align}
and
\begin{align}
\begin{split}
\label{eqn-PutCall-EarlyExercisePremium-J2}
& \tilde{V}^P_{J_2}(t,\tilde{s},v)\\
& \qquad = \sum_{m=0}^\infty \sum_{n=0}^\infty \frac{\tilde{\lambda}_1^m \tilde{\lambda}_2^n}{m!n!}\E_{\hat{\Q}}^{(m,n)}\Bigg[\frac{1}{2\pi}\int_0^T \int_0^\infty (u-t)^{m+n}e^{-(\tilde{\lambda}_1+\tilde{\lambda}_2)(u-t)}\\
& \hspace{25pt} \times \int_0^\infty G_2(y) \int_{\ln A(u,v_u)}^{\ln A(u,v_u)+y} \left(V^A(u,e^{x_u-y},v_u)-\left(e^{-q_1 u}e^{x_u-y}-e^{-q_2 u}\right)\right)\\
& \hspace{25pt} \times \int_{-\infty}^\infty e^{i\phi x_u}\exp\left\{-i\phi\ln\left(\tilde{s}e^{-(\tilde{\lambda}_1\tilde{\kappa}_1+\tilde{\lambda}_2\tilde{\kappa}_2^-)(u-t)}e^{\Upsilon_{1,m}-\Upsilon_{2,n}}\right)\right\}\\
& \hspace{25pt} \times h(u-t,\phi,v;v_u)\dif\phi \dif x_u \dif y \dif v_u \dif u \Bigg].
\end{split}
\end{align}
\end{prop}

\begin{proof}
See Appendix \ref{app-prop-PutCall-EarlyExercisePremium-IntegralRepresentation-Proof}.
\end{proof}

Because of the possibility of jumps in the prices of the two underlying assets, the early exercise premium, as seen from its jump components, remains dependent on the (yet unknown) discounted American exchange option price $\tilde{V}^A(t,\tilde{s},v)$. The early exercise premium components also require knowledge of the unknown critical asset yield ratio $B(t,v)$ for the term $A(t,v) = B(t,v)e^{(q_1-q_2)t}$ that appears in equations \eqref{eqn-PutCall-EarlyExercisePremium-Diffusion}, \eqref{eqn-PutCall-EarlyExercisePremium-J1}, and \eqref{eqn-PutCall-EarlyExercisePremium-J2}. In this regard, we require another equation that, when coupled with equation \eqref{eqn-PutCall-EarlyExerciseRepresentation}, characterizes both $\tilde{V}^A$ and $B(t,v)$.

\begin{prop}
\label{prop-PutCall-LinkedSystem}
The discounted American exchange option $\tilde{V}^A(t,\tilde{s},v)$ and the critical asset yield ratio $B(t,v)$ are the solution of the linked system of integral equations
\begin{align}
\begin{split}
\tilde{V}^A(t,\tilde{s},v) & = \tilde{V}(t,\tilde{s},v)+\tilde{V}^P(t,\tilde{s},v)\\
e^{-q_1 t}\left(A(t,v)-e^{(q_1-q_2)t}\right) & = \tilde{V}(t,A(t,v),v)+\tilde{V}^P(t,A(t,v),v),
\end{split}
\end{align}
where $A(t,v) = B(t,v)e^{(q_1-q_2)t}$, $\tilde{V}(t,\tilde{s},v)$ is the price of the European exchange option given in Proposition \ref{prop-PutCall-DiscEuExcOpPrice}, and $\tilde{V}^P(t,\tilde{s},v)$ is the early exercise premium given in Proposition \ref{prop-PutCall-EarlyExercisePremium-IntegralRepresentation}.
\end{prop}

\begin{proof}
The first equation is the early exercise representation for $\tilde{V}^A(t,\tilde{s},v)$ whereas the second equation is obtained by evaluating the first equation at $\tilde{s}=A(t,v)$ and invoking the value-matching condition \eqref{eqn-PutCall-ValueMatchingCondition-Vtilde} to rewrite the left-hand side. The linked system is supplemented with expressions for $\tilde{V}(t,\tilde{s},v)$ and $\tilde{V}^P(t,\tilde{s},v)$ given in Propositions \ref{prop-PutCall-DiscEuExcOpPrice} and \ref{prop-PutCall-EarlyExercisePremium-IntegralRepresentation}, respectively.
\end{proof}

The discounted American exchange option price $\tilde{V}^A(t,\tilde{s},v)$, once determined, is measured in units of the second asset yield process. To obtain the nominal price, $\tilde{V}^A(t,\tilde{s},v)$ is simply multipled by $S_2(t)e^{-q_2 t}$.

While our analysis covers only the application of the change-of-num\'eraire technique to reduce the dimensionality of the problem and producing integral representations of the option prices and the early exercise premium, our results may be linked to existing literature that grapple with the numerical aspects of the option pricing problem. Most notably, the method-of-lines (MOL) and the component-wise splitting approaches used by \citet{Chiarella-2009} may be used to solve the IPDE obtained in Proposition \ref{prop-PutCall-DiscountedEuExcOpPrice} and the corresponding free-boundary problem \eqref{eqn-PutCall-ValueMatchingCondition-Vtilde} to \eqref{eqn-PutCall-BoundaryConditions-Vtilde} for the American exchange option since the structure of the IPDE operator we defined in our analysis is similar to that tackled by \citet{Chiarella-2009}. Our IPDE, however, slightly differs due to the presence of a second integral component corresponding to the jumps in the second asset price. One may also extend the numerical methods explored by \citet{Adolfsson-2013} who priced single-asset American options under stochastic volatility dynamics and by \citet{Kang-2014} who priced American calls under stochastic volatility and interest rates. Moreover, the transition density function that we obtained may also be used in Monte Carlo simulation schemes to price options for which an analytical representation may not be readily available. Lastly, numerical integration schemes may be used to obtain option prices from our integral representations, provided that values for the model parameters are available.

\section{Summary and Conclusion}
\label{sec-PutCall-Conclusion}

We considered the problem of pricing European and American exchange options when the underlying asset prices are modelled with jump-diffusion dynamics and a common Heston-type stochastic volatility process. Our results are thus extensions of those obtained by \citet{Margrabe-1978} and \citet{Bjerskund-1993} for European and American exchange options under pure diffusion dynamics and those by \citet{CheangChiarella-2011} who analyzed exchange options under jump-diffusion dynamics. 

The results of this paper complement those presented by \citet{CheangGarces-2019}. In the SVJD model of \citet{CheangGarces-2019}, the assumption that asset prices were uncorrelated (among other simplifying assumptions on the correlation structure of the model) had to be enforced to obtain analytical representation of exchange option prices. In this paper, we used a proportional stochastic volatility jump-diffusion model to include the possibility of correlation between asset returns in our representation of exchange option prices. We showed that the representations we obtained in this analysis are of similar form to those obtained by \citet{CheangGarces-2019}, as well as those obtained in earlier studies by \citet{Margrabe-1978}, \citet{Fischer-1978}, and \citet{CheangChiarella-2011}.

Given the proportional SVJD model specification, the two-dimensional exchange option pricing problem was reduced to a one-dimensional call option pricing problem by taking the second asset yield process as the num\'eraire, following the put-call transformation method described by \citet{Bjerskund-1993}. We defined the equivalent martingale measure $\hat{\Q}$, corresponding to our choice of num\'eraire, which was then used to express the exchange option prices as expectations under $\hat{\Q}$. Then, with usual martingale arguments, we derived the IPDE for the European option price and the Kolmogorov backward equation for the joint transition density function of $\tilde{S}(t)$ and $v(t)$.

Having reduced the problem to pricing a call option on the asset yield ratio $\tilde{S}(t)$, which also has stochastic volatility and jump-diffusion dynamics, we then adapted the methodology of \citet{CheangChiarellaZiogas-2013} to obtain an early exercise representation for the American exchange option via probabilistic arguments and to solve for the joint transition density function via Fourier and Laplace transforms. Equipped with the joint transition density function, we then obtained integral representations for the price of the European exchange option and the early exercise premium and a linked system of integral equations for the American exchange option price and the unknown early exercise boundary associated to the American option. We also analyzed the limiting behavior of the early exercise boundary immediately before maturity and how the dividend yields and the jump parameters affected this limit.

With the simplification of the exchange option pricing problem, our results can then be linked to a number of existing numerical and simulation techniques for single-asset options \citep[see e.g.][]{Adolfsson-2013, Chiarella-2009} to obtain option prices given the model parameters.

\section*{Acknowledgments}

The first author is supported by a Research Training Program International (RPTi) scholarship awarded by the Australian Commonwealth Government and by a Faculty Development Grant from the Loyola Schools of Ateneo de Manila University.

\section*{Disclosure Statement}

The authors report no potential conflict of interest arising from the results of this paper.

\bibliographystyle{tfcad}
\bibliography{PutCallTransformationExchangeOption_arXiv}

\appendix

\section{Fourier Transform of the Transition Density Function IPDE}
\label{app-prop-PutCall-FourierTransformedPDE-Proof}

A straightforward application of equation \eqref{eqn-PutCall-FourierTransform} yields $$\scrF_x\left\{\pder[H]{t}\right\} = \pder[\hat{H}]{t} \qquad \scrF_x\left\{\pder[H]{v}\right\} = \pder[\hat{H}]{v} \qquad \scrF_x\left\{\pder[^2H]{v^2}\right\} = \pder[^2\hat{H}]{v^2}.$$ Furthermore, with Assumption \ref{asp-PutCall-FourierTransformAssumption}, integration by parts gives us $$\scrF_x\left\{\pder[H]{x}\right\} = -i\phi\hat{H}, \qquad \scrF_x\left\{\pder[^2H]{x^2}\right\} = -\phi^2\hat{H}, \qquad \scrF_x\left\{\pder[^2H]{x\partial v}\right\} = -i\phi\pder[\hat{H}]{v}.$$ The Fourier transform of the first expectation in equation \eqref{eqn-PutCall-KBE-ChangeofVariable} is the Fourier transform of a convolution-type integral and is calculated as follows:
\begin{align*}
\scrF_x\left\{\E_{\hat{\Q}}^{Y_1}\left[H(t,x+Y_1,v)\right]\right\}
	& = \int_{-\infty}^\infty\int_{-\infty}^\infty e^{i\phi x}H(t,x+y,v)G_1(y)\dif y\dif x\\
	& = \int_{-\infty}^\infty e^{-i\phi y}G_1(y) \left(\int_{-\infty}^\infty e^{i\phi z}H(t,z,v)\dif z\right) \dif y\\
	& = \varphi_1(\phi)\hat{H}(t,\phi,v).
\end{align*}
A similar calculation shows that $\scrF_x\{\E_{\hat{\Q}}^{Y_2}[H(t,x-Y_2,v)]\} = \varphi_2(-\phi)\hat{H}(t,\phi,v).$

Thus, the Fourier transform of equation \eqref{eqn-PutCall-KBE-ChangeofVariable} is
\begin{align*}
0	& = \pder[\hat{H}]{t}+i\phi\left(\tilde{\lambda}_1\tilde{\kappa}_1+\tilde{\lambda}_2\tilde{\kappa}_2^-+\frac{1}{2}\sigma^2 v\right)\hat{H} + \left[\xi\eta-(\xi+\Lambda)v\right]\pder[\hat{H}]{v}\\
	& \qquad - \frac{1}{2}\sigma^2 v\phi^2 \hat{H} + \frac{1}{2}\omega^2 v\pder[^2\hat{H}]{v^2}-\omega\left(\sigma_1\rho_1-\sigma_2\rho_2\right)v i\phi\pder[\hat{H}]{v}\\
	& \qquad + \tilde{\lambda}_1\left(\varphi_1(\phi)-1\right)\hat{H}+\tilde{\lambda}_2\left(\varphi_2(-\phi)-1\right)\hat{H}.
\end{align*}
Factoring the above expression and using the notation introduced in equation \eqref{eqn-PutCall-FourierTransformedPDE-Parameters} yields equation \eqref{eqn-PutCall-FourierTransformedPDE}.

The associated terminal condition is obtained by a straightforward application of equation \eqref{eqn-PutCall-FourierTransform} and the definition of the Dirac-delta function.

\section{Laplace Transform of the Transformed PDE}
\label{app-prop-PutCall-LaplaceTransformedPDE-Proof}

Given that Assumption \ref{asp-PutCall-LaplaceTransformAssumption} holds, we find that
\begin{align*}
\scrL_v\left\{\pder[\hat{H}]{t}\right\} & = \pder[\bar{H}]{t}\\
\scrL_v\left\{\left(\frac{1}{2}\varepsilon v-i\phi\Psi\right)\hat{H}\right\} & = -\frac{1}{2}\varepsilon\pder[\bar{H}]{\vartheta}-i\phi\Psi\bar{H}\\
\scrL_v\left\{(\alpha-\Theta v)\pder[\hat{H}]{v}\right\} & = -\alpha\hat{H}(t,\phi,0)+(\alpha\vartheta+\Theta)\bar{H}+\Theta\vartheta\pder[\bar{H}]{\vartheta}\\
\scrL_v\left\{\frac{1}{2}\omega^2 v\pder[^2\hat{H}]{v^2}\right\} & = \frac{1}{2}\omega^2\hat{H}(t,\phi,0)-\frac{1}{2}\omega^2\vartheta^2\pder[\bar{H}]{\vartheta}-\omega^2\vartheta\bar{H}.
\end{align*}
The Laplace transform of equation \eqref{eqn-PutCall-FourierTransformedPDE} is therefore
\begin{align*}
0 & = \pder[\bar{H}]{t} -\frac{1}{2}\varepsilon\pder[\bar{H}]{\vartheta}-i\phi\Psi\bar{H} -\alpha\hat{H}(t,\phi,0)+(\alpha\vartheta+\Theta)\bar{H}+\Theta\vartheta\pder[\bar{H}]{\vartheta}\\
	& \qquad + \frac{1}{2}\omega^2\hat{H}(t,\phi,0)-\frac{1}{2}\omega^2\vartheta^2\pder[\bar{H}]{\vartheta}-\omega^2\vartheta\bar{H}.
\end{align*}
Factoring and defining $f(t)\equiv (\omega^2/2-\alpha)\hat{H}(t,\phi,0)$ yields equation \eqref{eqn-PutCall-LaplaceTransformedPDE}.

Note that, at this point, $\hat{H}(t,\phi,0)$ is unknown but it has to be determined so that the solution $\bar{H}$ to equation \eqref{eqn-PutCall-LaplaceTransformedPDE} is finite for all $\vartheta>0$. A sufficient condition for this is equation \eqref{eqn-PutCall-LaplaceTransformedPDE-FinitenessCondition} \citep{CheangChiarellaZiogas-2013}.

The terminal condition \eqref{eqn-PutCall-LaplaceTransformedPDE-TerminalCondition} is a consequence of the definition \eqref{eqn-PutCall-LaplaceTransform} and the Dirac-delta function.

\section{Recovering the Transition Density Function}
\label{app-prop-PutCall-TransitionDensityFunction2-Proof}

Using the inversion formula \eqref{eqn-PutCall-InverseFourierTransform} and the expressions for $\Psi$ and $\hat{H}(t,\phi,v)$ in equations \eqref{eqn-PutCall-FourierTransformedPDE-Parameters} and \eqref{eqn-PutCall-InverseLaplaceTransform}, respectively, $H(t,x,v)$ is given by
\begin{align*}
H(t,x,v)
	& = \frac{1}{2\pi}\int_{-\infty}^\infty \exp\left\{-i\phi\left[x-x_T-\tilde{\lambda}_1\tilde{\kappa}_1(T-t)-\tilde{\lambda}_2\tilde{\kappa}_2^-(T-t)\right]\right\}\\
	& \qquad \times \exp\left\{\tilde{\lambda}_1(T-t)[\varphi_1(\phi)-1]\right\} \exp\left\{\tilde{\lambda}_2(T-t)[\varphi_2(-\phi)-1]\right\}\\
	& \qquad \times h(T-t,\phi,v;v_T)\dif\phi,
\end{align*}
where $h$ is the function given by equation \eqref{eqn-PutCall-TransitionDensityFunction2-h}.

We note that $$e^{\tilde{\lambda}_1(T-t)\varphi_1(\phi)} = \sum_{m=0}^\infty \frac{(\tilde{\lambda}_1(T-t))^m}{m!}[\varphi_1(\phi)]^m.$$ If $Y_{1,1},\dots,Y_{1,m}$ are i.i.d. random variables with density function $G_1(y)$, then $[\varphi_1(\phi)]^m = \E_{\hat{\Q}}^{(m)}[e^{-i\phi\Upsilon_{1,m}}]$, where $\Upsilon_{1,m}=\sum_{k=1}^m Y_{1,k}$, since $\varphi_1(\cdot)$ is the characteristic function of each of the $Y_{1,k}$'s. Here, the expectation operator $\E_{\hat{\Q}}^{(m)}$ acts only on $\Upsilon_{1,m}$. It follows that $$\exp\left\{\tilde{\lambda}_1(T-t)[\varphi_1(\phi)-1]\right\} = e^{-\tilde{\lambda}_1(T-t)}\sum_{m=0}^\infty\frac{(\tilde{\lambda}_1(T-t))^m}{m!}\E_{\hat{\Q}}^{(m)}\left[e^{-i\phi\Upsilon_{1,m}}\right].$$ A similar analysis calculation yields $$\exp\left\{\tilde{\lambda}_2(T-t)[\varphi_2(-\phi)-1]\right\} = e^{-\tilde{\lambda}_2(T-t)}\sum_{n=0}^\infty\frac{(\tilde{\lambda}_2(T-t))^n}{n!}\E_{\hat{\Q}}^{(n)}\left[e^{i\phi\Upsilon_{2,n}}\right],$$ where $\Upsilon_{2,n}=\sum_{l=1}^n Y_{2,l}$, $\{Y_{2,1},\dots,Y_{2,n}\}$ is a sample of i.i.d. random variables from a population with density function $G_2(y)$, and $\E_{\hat{\Q}}^{(n)}$ is the expectation operator acting only on $\Upsilon_{2,n}$. We then obtain
\begin{align}
\begin{split}
\label{eqn-PutCall-TransitionDensityFunction2-LogProcess}
H(t,x,v)
	& = \sum_{m=0}^\infty\sum_{n=0}^\infty \frac{(\tilde{\lambda}_1(T-t))^m(\tilde{\lambda}_2(T-t))^n e^{-(\tilde{\lambda}_1+\tilde{\lambda}_2)(T-t)}}{m!n!}\\
	& \qquad \times \E_{\hat{\Q}}^{(m,n)}\Bigg[\frac{1}{2\pi}\int_{-\infty}^\infty \exp\left\{-i\phi\left[x-x_T-\tilde{\lambda}_1\tilde{\kappa}_1(T-t)-\tilde{\lambda}_2\tilde{\kappa}_2^-(T-t)\right]\right\}\\
	& \hspace{75pt} \times \exp\left\{-i\phi\left(\Upsilon_{1,m}-\Upsilon_{2,n}\right)\right\} h(T-t,\phi,v;v_T)\dif\phi\Bigg].
\end{split}
\end{align}
Here, $\E_{\hat{\Q}}^{(m,n)}$ is the expectation acting only on $\Upsilon_{1,m}$ and $\Upsilon_{2,n}$.

Reversing the substitutions $x=\ln\tilde{s}$ and $x_T=\ln\tilde{s}_T$ and combining the resulting exponent under one natural logarithm yields the desired result.

\section{Integral Representation of the Discounted European Exchange Option Price}
\label{app-prop-PutCall-DiscEuExcOpPrice-Proof}

From equation \eqref{eqn-PutCall-EuExcOpPrice} and Lemma \ref{lem-PutCall-HTDFIntegral}, $\tilde{V}(t,\tilde{s},v)$ may be written as
\begin{align*}
\tilde{V}(t,\tilde{s},v)
& = e^{-q_1 T}\sum_{m=0}^\infty \sum_{n=0}^\infty \frac{(\tilde{\lambda}_1(T-t))^m (\tilde{\lambda}_2(T-t))^n e^{-(\tilde{\lambda}_1+\tilde{\lambda}_2)(T-t)}}{m! n!}\\
& \qquad \times \E_{\hat{\Q}}^{(m,n)}\Bigg[\frac{1}{2\pi}\int_{-\infty}^\infty \int_{(q_1-q_2)T}^\infty \left(e^{x_T}-e^{(q_1-q_2)T}\right)e^{i\phi x_T}\\
& \hspace{75pt} \times f\left(T-t,\tilde{s}e^{-(\tilde{\lambda}_1\tilde{\kappa}_1+\tilde{\lambda}_2\tilde{\kappa}_2^-)(T-t)}e^{\Upsilon_{1,m}-\Upsilon_{2,n}},v;-\phi\right)\dif x_T\dif\phi\Bigg].
\end{align*}
To simplify notation, temporarily set $$f(-\phi)\equiv f(T-t,\tilde{s}e^{-(\tilde{\lambda}_1\tilde{\kappa}_1+\tilde{\lambda}_2\tilde{\kappa}_2^-)(T-t)}e^{\Upsilon_{1,m}-\Upsilon_{2,n}},v;-\phi)$$ and $K \equiv (q_1-q_2)T$. Furthermore, let $\mathbb{I}$ denote the double integral inside the above expectation. It follows that $$\mathbb{I} = \int_{-\infty}^\infty \int_K^\infty e^{i(\phi-i)x_T}f(-\phi)\dif x_T\dif\phi - e^K\int_{-\infty}^\infty \int_K^\infty e^{i\phi x_T}f(-\phi)\dif x_T\dif\phi.$$

Denote by $\mathbb{I}_1$ and $\mathbb{I}_2$ the first and second terms, respectively, of $\mathbb{I}$. Integrating with respect to $x_T$, we have
\begin{align*}
\mathbb{I}_1
	& = \int_{-\infty}^\infty f(-\phi)\int_K^\infty e^{i(\phi-i)x_T}\dif x_T\dif\phi = \lim_{b\to\infty} \int_{-\infty}^\infty f(-\phi)\left[\frac{e^{i(\phi-i)b}-e^{i(\phi-i)K}}{i(\phi-i)}\right]\dif\phi\\
	& = \lim_{b\to\infty}\int_{-\infty}^\infty f(-u-i)\left[\frac{e^{iub}-e^{iuK}}{iu}\right]\dif u = \lim_{b\to\infty}\int_{-\infty}^\infty f(\phi-i)\left[\frac{e^{-i\phi K}-e^{-i\phi b}}{i\phi}\right]\dif\phi\\
	& = \int_0^\infty\frac{f(\phi-i)e^{-i\phi K}-f(-\phi-i)e^{i\phi K}}{i\phi}\dif\phi\\
	& \qquad -\lim_{b\to\infty}\int_0^\infty\frac{f(\phi-i)e^{-i\phi b}-f(-\phi-i)e^{i\phi b}}{i\phi}\dif\phi,
\end{align*}
where the last equality is a result of splitting the domain of integration on $\phi$. Observe that
{\small\begin{align*}
f(\phi-i)
	& = f(T-t,\tilde{s}e^{-(\tilde{\lambda}_1\tilde{\kappa}_1+\tilde{\lambda}_2\tilde{\kappa}_2^-)(T-t)}e^{\Upsilon_{1,m}-\Upsilon_{2,n}},v;\phi-i)\\
	& = \tilde{s}e^{-(\tilde{\lambda}_1\tilde{\kappa}_1+\tilde{\lambda}_2\tilde{\kappa}_2^-)(T-t)}e^{\Upsilon_{1,m}-\Upsilon_{2,n}} f_1\left(T-t,\tilde{s}e^{-(\tilde{\lambda}_1\tilde{\kappa}_1+\tilde{\lambda}_2\tilde{\kappa}_2^-)(T-t)}e^{\Upsilon_{1,m}-\Upsilon_{2,n}},v;\phi\right),
\end{align*}}
where $f_1$ and its components are defined in equation \eqref{eqn-PutCall-DiscEuExcOpPrice-Parameters2}. Thus, we may write
{\small\begin{align*}
\mathbb{I}_1
	& = \tilde{s}e^{-(\tilde{\lambda}_1\tilde{\kappa}_1+\tilde{\lambda}_2\tilde{\kappa}_2^-)(T-t)}e^{\Upsilon_{1,m}-\Upsilon_{2,n}}\\
	& \qquad \times \Bigg[\int_0^\infty\frac{f_1(\phi)e^{-i\phi K}-f_1(-\phi)e^{i\phi K}}{i\phi}\dif\phi - \lim_{b\to\infty}\int_0^\infty\frac{f_1(\phi)e^{-i\phi b}-f_1(-\phi)e^{i\phi b}}{i\phi}\dif\phi\Bigg],
\end{align*}}
where $f_1(\phi) \equiv f_1(T-t,\tilde{s}e^{-(\tilde{\lambda}_1\tilde{\kappa}_1+\tilde{\lambda}_2\tilde{\kappa}_2^-)(T-t)}e^{\Upsilon_{1,m}-\Upsilon_{2,n}},v;\phi)$. 

The limit above may be evaluated using the results in \citet{Shephard-1991b}. Let $F_1(x)$ be the cumulative distribution function of a random variable whose mean exists and whose characteristic function is $f_1(\phi)$. Then by \citet[Theorem 3]{Shephard-1991b}, we have $$F_1(b) = \frac{1}{2}-\frac{1}{2\pi}\int_0^\infty \frac{f_1(\phi)e^{-i\phi b}-f_1(-\phi)e^{i\phi b}}{i\phi}\dif \phi.$$ Since $F_1(b)\to 1$ as $b\to\infty$, it follows that $$\mathbb{I}_1 = \tilde{s}e^{-(\tilde{\lambda}_1\tilde{\kappa}_1+\tilde{\lambda}_2\tilde{\kappa}_2^-)(T-t)}e^{\Upsilon_{1,m}-\Upsilon_{2,n}}\left[\pi+\int_0^\infty\frac{f_1(\phi)e^{-i\phi K}-f_1(-\phi)e^{i\phi K}}{i\phi}\dif\phi\right].$$

The same calculations apply for evaluating $\mathbb{I}_2$. Evaluating the integral with respect to $x_T$ and splitting the domain of integration in $\phi$, we have
\begin{align*}
\mathbb{I}_2
	& = \lim_{b\to\infty}\int_{-\infty}^\infty f(-\phi)\left[\frac{e^{i\phi b}-e^{i\phi K}}{i\phi}\right]\dif\phi\\
	& = \int_0^\infty\frac{f(\phi)e^{-i\phi K}-f(-\phi)e^{i\phi K}}{i\phi}\dif\phi - \lim_{b\to\infty}\int_0^\infty\frac{f(\phi)e^{-i\phi b}-f(-\phi)e^{i\phi b}}{i\phi}\dif\phi.
\end{align*}
Using \citet[Theorem 3]{Shephard-1991b} to evaluate the limit, we find that $$\mathbb{I}_2 = \pi+\int_0^\infty\frac{f(\phi)e^{-i\phi K}-f(-\phi)e^{i\phi K}}{i\phi}\dif\phi.$$

Using these results to rewrite the double integral in $\tilde{V}(t,\tilde{s},v)$, putting back $(q_1-q_2)T$ in place of $K$, and writing the resulting expression in terms of $P_1^E$ and $P_2^E$ yield the expression presented in the proposition.

\section{Integral Representation of the Early Exercise Premium}
\label{app-prop-PutCall-EarlyExercisePremium-IntegralRepresentation-Proof}

Equation \eqref{eqn-PutCall-earlyExercisePremium4} for the early exercise premium is written such that $\tilde{V}^P_D$, $\tilde{V}^P_{J_1}$, and $\tilde{V}^P_{J_2}$ denote the first, second, and third expectations, respectively, that appear in equation \eqref{eqn-PutCall-EarlyExercisePremium3}. These terms may be interpreted as the diffusion and jump components, respectively, of the early exercise premium. In each term, the indicator function $\vm{1}(\calA(u))$ appears, where $u\in[t,T]$. In terms of $x_u$ and $v_u$, the event $\calA(u)$ may be written as $\calA(u) = \{x_u \geq \ln A(u,v_u)\}$.

\textbf{\underline{Evaluating $\tilde{V}^P_D$:}} From equation \eqref{eqn-PutCall-EarlyExercisePremium3}, $\tilde{V}^P_D$ is given by $$\tilde{V}^P_D(t,\tilde{s},v) = \E_{\hat{Q}}\Bigg[\left.\int_t^T e^{-q_1 u}\left(q_1 \tilde{s}_u-q_2 e^{(q_1-q_2)u}\right)\vm{1}(\calA(u))\right|\tilde{S}(t) = \tilde{s}, v(t) = v\Bigg],$$ and in terms of the transition density function $H$ in equation \eqref{eqn-PutCall-TransitionDensityFunction2-LogProcess}, we have $$\tilde{V}^P_D(t,\tilde{s},v) = \int_t^T \int_0^\infty \int_{\ln A(u,v_u)}^\infty e^{-q_1 u}\left(q_1 e^{x_u}-q_2 e^{K(u)}\right)H(u,x_u,v_u;t,\ln\tilde{s},v)\dif x_u\dif v_u \dif t,$$ where $K(u)\equiv (q_1-q_2)u$. Using equation \eqref{eqn-PutCall-TransitionDensityFunction2-LogProcess}, we find that
\begin{align*}
\tilde{V}^P_D(t,\tilde{s},v)
& = \sum_{m=0}^\infty \sum_{n=0}^\infty \frac{\tilde{\lambda}_1^m\tilde{\lambda}_2^n}{m!n!}\E_{\hat{\Q}}^{(m,n)}\Bigg[\int_t^T \int_0^\infty (u-t)^{m+n} e^{-(\tilde{\lambda}_1+\tilde{\lambda}_2)(u-t)}e^{-q_1 u}\\
& \qquad \times \frac{1}{2\pi}\int_{\ln A(u,v_u)}^\infty \left(q_1 e^{x_u}-q_2 e^{K(u)}\right)\int_{-\infty}^\infty e^{i\phi x_u}\\
& \qquad \times \exp\left\{-i\phi\ln\left(\tilde{s}e^{-(\tilde{\lambda}_1\tilde{\kappa}_1+\tilde{\lambda}_2\tilde{\kappa}_2^-)(u-t)}e^{\Upsilon_{1,m}-\Upsilon_{2,n}}\right)\right\}\\
& \qquad \times h(u-t,\phi,v;v_u)\dif\phi \dif x_u \dif v_u \dif u \Bigg].
\end{align*}
In terms of the function $f_2$ introduced in the proposition, we can write
\begin{align}
\begin{split}
\label{eqn-PutCall-EarlyExercisePremium-Diffusion2}
\tilde{V}^P_D(t,\tilde{s},v)
& = \sum_{m=0}^\infty \sum_{n=0}^\infty \frac{\tilde{\lambda}_1^m\tilde{\lambda}_2^n}{m!n!}\E_{\hat{\Q}}^{(m,n)}\Bigg[\int_t^T \int_0^\infty (u-t)^{m+n} e^{-(\tilde{\lambda}_1+\tilde{\lambda}_2)(u-t)}e^{-q_1 u}\\
& \qquad \times \frac{1}{2\pi}\int_{-\infty}^\infty f_2(-\phi)\int_{\ln A(u,v_u)}^\infty \left(q_1 e^{x_u}-q_2 e^{K(u)}\right) e^{i\phi x_u}\dif x_u \dif \phi \dif v_u \dif u\Bigg],
\end{split}
\end{align}
where we set $f_2(-\phi) \equiv f_2(u-t,\tilde{s}e^{-(\tilde{\lambda}_1\tilde{\kappa}_1+\tilde{\lambda}_2\tilde{\kappa}_2^-)(u-t)}e^{\Upsilon_{1,m}-\Upsilon_{2,n}},y;-\phi,v_u)$.

We next integrate the innermost double integral, which we denote by $\mathbb{I}^*$, using the techniques employed in the proof of Proposition \ref{prop-PutCall-DiscEuExcOpPrice}. To this end, we make a secondary representation $\mathbb{I}^* = q_1 \mathbb{I}^*_1-q_2 e^{K(u)}\mathbb{I}_2^*$, where
\begin{align*}
\mathbb{I}_1^* & = \int_{-\infty}^\infty f_2(-\phi) \int_{\ln A(u,v_u)}^\infty e^{i(\phi-i)x_u}\dif x_u \dif\phi\\
\mathbb{I}_2^* & = \int_{-\infty}^\infty f_2(-\phi) \int_{\ln A(u,v_u)}^\infty e^{i\phi x_u}\dif x_u \dif\phi.
\end{align*}

Evaluating the integral with respect to $x_u$, $\mathbb{I}_1^*$ can be written as
\begin{align*}
\mathbb{I}_1^*
	& = \lim_{b\to\infty}\int_{-\infty}^\infty f_2(-\phi)\left[\frac{e^{i(\phi-i)b}-e^{i(\phi-i)\ln A(u,v_u)}}{i(\phi-i)}\right]\dif\phi\\
	& = \lim_{b\to\infty}\int_{-\infty}^\infty f_2(\phi-i)\left[\frac{e^{-i\phi\ln A(u,v_u)}-e^{-i\phi b}}{i\phi}\right]\dif \phi\\
	& = \int_0^\infty \frac{f_2(\phi-i)e^{-i\phi\ln A(u,v_u)}-f_2(-\phi-i)e^{i\phi\ln A(u,v_u)}}{i\phi}\dif\phi\\
	& \qquad - \lim_{b\to\infty}\int_0^\infty \frac{f_2(\phi-i)e^{-i\phi b}-f_2(-\phi-i)e^{i\phi b}}{i\phi}\dif\phi.
\end{align*}
Note that
\begin{align*}
f_2(\phi-i)
	& \equiv f_2\left(u-t,\tilde{s}e^{-(\tilde{\lambda}_1\tilde{\kappa}_1+\tilde{\lambda}_2\tilde{\kappa}_2^-)(u-t)}e^{\Upsilon_{1,m}-\Upsilon_{2,n}},y;\phi-i,v_u\right)\\
	& = \tilde{s}e^{-(\tilde{\lambda}_1\tilde{\kappa}_1+\tilde{\lambda}_2\tilde{\kappa}_2^-)(u-t)}e^{\Upsilon_{1,m}-\Upsilon_{2,n}} \\
	& \qquad \times f_{21}\left(u-t,\tilde{s}e^{-(\tilde{\lambda}_1\tilde{\kappa}_1+\tilde{\lambda}_2\tilde{\kappa}_2^-)(u-t)}e^{\Upsilon_{1,m}-\Upsilon_{2,n}},v;\phi,v_u\right),
\end{align*}
where $f_{21}$ is defined in the proposition. For ease of notation, let $f_{21}(\phi)$ be shorthand for $f_{21}(u-t,\tilde{s}e^{-(\tilde{\lambda}_1\tilde{\kappa}_1+\tilde{\lambda}_2\tilde{\kappa}_2^-)(u-t)}e^{\Upsilon_{1,m}-\Upsilon_{2,n}},v;\phi,v_u)$. With this representation, $\mathbb{I}_1^*$ can be written as
\begin{align*}
\mathbb{I}_1^*
	& = \tilde{s}e^{-(\tilde{\lambda}_1\tilde{\kappa}_1+\tilde{\lambda}_2\tilde{\kappa}_2^-)(u-t)}e^{\Upsilon_{1,m}-\Upsilon_{2,n}}\\
	& \qquad \times \Bigg[\int_0^\infty \frac{f_{21}(\phi)e^{-i\phi\ln A(u,v_u)}-f_{21}(-\phi)e^{i\phi\ln A(u,v_u)}}{i\phi}\dif\phi\\
	& \qquad \qquad -\lim_{b\to\infty}\int_0^\infty\frac{f_{21}(\phi)e^{-i\phi b}-f_{21}(-\phi)e^{i\phi b}}{i\phi}\dif\phi\Bigg]\\
	& = \tilde{s}e^{-(\tilde{\lambda}_1\tilde{\kappa}_1+\tilde{\lambda}_2\tilde{\kappa}_2^-)(u-t)}e^{\Upsilon_{1,m}-\Upsilon_{2,n}}\\
	& \qquad \times \Bigg[\pi+\int_0^\infty \frac{f_{21}(\phi)e^{-i\phi\ln A(u,v_u)}-f_{21}(-\phi)e^{i\phi\ln A(u,v_u)}}{i\phi}\dif\phi\Bigg],
\end{align*}
where the limit was evaluated using Theorem 3 of \citet{Shephard-1991b}.

Similar calculations apply in simplifying $\mathbb{I}_2^*$. Integration with respect to $x_u$ yields
\begin{align*}
\mathbb{I}_2^*
	& = \lim_{b\to\infty}\int_{-\infty}^\infty f_2(-\phi)\left[\frac{e^{i\phi b}-e^{i\phi\ln A(u,v_u)}}{i\phi}\right]\dif\phi\\
	& = \int_0^\infty \frac{f_2(\phi)e^{-i\phi\ln A(u,v_u)}-f_2(-\phi)e^{i\phi\ln A(u,v_u)}}{i\phi}\dif\phi\\
	& \qquad - \lim_{b\to\infty}\int_0^\infty \frac{f_2(\phi)e^{-i\phi b}-f_2(-\phi)e^{i\phi b}}{i\phi}\dif\phi.
\end{align*}
By Theorem 3 of \citet{Shephard-1991b}, the limit reduces to $-\pi$, giving us $$\mathbb{I}_2^* = \pi+\int_0^\infty \frac{f_2(\phi)e^{-i\phi\ln A(u,v_u)}-f_2(-\phi)e^{i\phi\ln A(u,v_u)}}{i\phi}\dif\phi.$$

From our calculations, we now find that $\mathbb{I}^*$ is given by
\begin{align*}
\mathbb{I^*}
	& = q_1\tilde{s}e^{-(\tilde{\lambda}_1\tilde{\kappa}_1+\tilde{\lambda}_2\tilde{\kappa}_2^-)(u-t)}e^{\Upsilon_{1,m}-\Upsilon_{2,n}}\\
	& \qquad \times \Bigg(\pi+\int_0^\infty \frac{f_{21}(\phi)e^{-i\phi\ln A(u,v_u)}-f_{21}(-\phi)e^{i\phi\ln A(u,v_u)}}{i\phi}\dif\phi\Bigg)\\
	& \qquad - q_2 e^{(q_1-q_2)u}\Bigg(\pi+\int_0^\infty \frac{f_2(\phi)e^{-i\phi\ln A(u,v_u)}-f_2(-\phi)e^{i\phi\ln A(u,v_u)}}{i\phi}\dif\phi\Bigg).
\end{align*}
Substituting this expression in place of the innermost double integral in equation \eqref{eqn-PutCall-EarlyExercisePremium-Diffusion2}, multiplying by $e^{-q_1 u}/(2\pi)$, and expressing the resulting expression in terms of $P_1^A$ and $P_2^A$ (as defined in the proposition) yields equation \eqref{eqn-PutCall-EarlyExercisePremium-Diffusion} for the diffusion component of the early exercise premium.

\textbf{\underline{Evaluating the jump components:}} From equation \eqref{eqn-PutCall-EarlyExercisePremium3} and with a change of variable from $\tilde{s}_u$ to $x_u$, $\tilde{V}^P_{J_1}$ can be expressed as
\begin{align*}
\tilde{V}^P_{J_1}(t,\tilde{s},v)
	& = \E_{\hat{\Q}}\Bigg[\int_t^T\int_{-\infty}^{\ln A(u,v_u)-x_u}\left(\tilde{V}^A(u,e^{x_u+y},v_u)-\left(e^{-q_1 u}e^{x_u+y}-e^{-q_2 u}\right)\right)\\
	& \qquad \times G_1(y)\vm{1}(\calA_u))\dif y\dif u\bigg| X(t) = \ln\tilde{s}, v(t) = v\Bigg],
\end{align*}
and in terms of the transition density function $H$, we have
{\small\begin{align*}
\tilde{V}^P_{J_1}(t,\tilde{s},v) 
& = \int_t^T\int_0^\infty\int_{\ln A(u,v_u)}^\infty \int_{-\infty}^{\ln A(u,v_u)-x_u}\left(\tilde{V}^A(u,e^{x_u+y},v_u)-\left(e^{-q_1 u}e^{x_u+y}-e^{-q_2 u}\right)\right)\\
& \qquad \times G_1(y) H(u,x_u,v_u;t,\ln\tilde{s},v)\dif y\dif x_u\dif v_u \dif u\\
& = \int_t^T \int_0^\infty \int_{-\infty}^0 \int_{\ln A(u,v_u)}^{\ln A(u,v_u)-y}\left(\tilde{V}^A(u,e^{x_u+y},v_u)-\left(e^{-q_1 u}e^{x_u+y}-e^{-q_2 u}\right)\right)\\
& \qquad \times G_1(y) H(u,x_u,v_u;t,\ln\tilde{s},v)\dif x_u \dif y \dif v_u \dif u,
\end{align*}}
where the last expression was obtained by changing the order of integration with respect to $x_u$ and $y$. Substitution of equation \eqref{eqn-PutCall-TransitionDensityFunction2-LogProcess} for $H$ in the equation above yields equation \eqref{eqn-PutCall-EarlyExercisePremium-J1} in the proposition.

We follow similar steps in evaluating the second jump component. From equation \eqref{eqn-PutCall-EarlyExercisePremium3} and in terms of $x_u$, $\tilde{V}^P_{J_1}$ can be written as
\begin{align*}
\tilde{V}^P_{J_2}(t,\tilde{s},v)
	& = \E_{\hat{\Q}}\Bigg[\int_t^T\int_{x_u-\ln A(u,v_u)}^\infty\left(\tilde{V}^A(u,e^{x_u+y},v_u)-\left(e^{-q_1 u}e^{x_u+y}-e^{-q_2 u}\right)\right)\\
	& \qquad \times G_2(y)\vm{1}(\calA_u))\dif y\dif u\bigg| X(t) = \ln\tilde{s}, v(t) = v\Bigg].
\end{align*}
We may equivalently write this in terms of the transition density function as
{\small\begin{align*}
\tilde{V}^P_{J_2}(t,\tilde{s},v)
& = \int_t^T\int_0^\infty\int_{\ln A(u,v_u)}^\infty \int_{x_u-\ln A(u,v_u)}^\infty\left(\tilde{V}^A(u,e^{x_u+y},v_u)-\left(e^{-q_1 u}e^{x_u+y}-e^{-q_2 u}\right)\right)\\
& \qquad \times G_2(y) H(u,x_u,v_u;t,\ln\tilde{s},v)\dif y\dif x_u\dif v_u \dif u\\
& = \int_t^T \int_0^\infty \int_0^\infty \int_{\ln A(u,v_u)}^{\ln A(u,v_u)+y}\left(\tilde{V}^A(u,e^{x_u+y},v_u)-\left(e^{-q_1 u}e^{x_u+y}-e^{-q_2 u}\right)\right)\\
& \qquad \times G_2(y) H(u,x_u,v_u;t,\ln\tilde{s},v) \dif x_u \dif y\dif v_u \dif u
\end{align*}}
Here, we changed the order of integration with respect to $x_u$ and $y$ to obtain the last expression. Replacing $H$ with its expression in equation \eqref{eqn-PutCall-TransitionDensityFunction2-LogProcess} gives us equation \eqref{eqn-PutCall-EarlyExercisePremium-J2} for the second jump component.

\end{document}